\documentclass[lettersize,journal]{IEEEtran}

\usepackage{tikz}

\usepackage[lined,ruled,commentsnumbered]{algorithm2e}

\usepackage{cite}
\usepackage[numbers,sort&compress]{}
\usepackage{amsmath,amssymb,amsfonts}
\usepackage{array}
\usepackage{graphicx}
\usepackage{makecell} 

\usepackage{textcomp}
\usepackage{xcolor}
\usepackage{booktabs}
\usepackage{multirow}
\usepackage{subfigure}
\usepackage{longtable}
\usepackage{threeparttable}
\usepackage{enumitem}
\usepackage{fontawesome}
\usepackage{amssymb}
\usepackage{soul}
\usepackage{textcomp}
\usepackage{xcolor}
\usepackage{array}
\usepackage{url}
\usepackage{float}
\usepackage{bm}
\usepackage{adjustbox}
\usepackage{epstopdf}

\usepackage{diagbox}

\usepackage{amsthm}
\usepackage{tabularx}
\usepackage{mathrsfs} 
\usepackage{soul}
\usepackage{cuted}   

\theoremstyle{plain} 
\newtheorem{theorem}{Theorem}[section] 
\newtheorem{lemma}[theorem]{Lemma}    

\usepackage{algpseudocode}

\usepackage{stfloats}
\usepackage[colorlinks=true, linkcolor=blue, citecolor=blue]{hyperref}

\usepackage{cleveref}
\begin{document}
\newcommand{\minitab}[2][l]{\begin{tabular}{#1}#2\end{tabular}}

\def\BibTeX{{\rm B\kern-.05em{\sc i\kern-.025em b}\kern-.08em
    T\kern-.1667em\lower.7ex\hbox{E}\kern-.125emX}}


\title{Backscatter Device-aided Integrated Sensing and Communication: A Pareto Optimization Framework}

\author{Yifan~Zhang,~\IEEEmembership{Graduate Student Member,~IEEE,}
        Yu Bai,~\IEEEmembership{Graduate Student Member,~IEEE,} Shuhao~Zeng,~\IEEEmembership{Member,~IEEE,}      
        Riku Jäntti,~\IEEEmembership{Senior Member,~IEEE,}
        Zheng~Yan,~\IEEEmembership{Fellow,~IEEE,}\\
        Christos Masouros,~\IEEEmembership{Fellow,~IEEE,}
        and Zhu~Han,~\IEEEmembership{Fellow,~IEEE.}
\thanks{Y.F. Zhang, Y. Bai, and R. Jäntti are with the Department of Information and Communications Engineering, Aalto University, Espoo, 02150, Finland. (email: yifan.1.zhang@aalto.fi, yu.bai@aalto.fi, riku.jantti@aalto.fi)}
\thanks{S.H. Zeng is with the Department of Electrical and Computer Engineering, Princeton University, NJ, USA (email: shuhao.zeng96@gmail.com).}%
\thanks{Z. Yan is with the State Key Lab of ISN, School of Cyber Engineering, Xidian University, Xi'an, Shaanxi, 710026 China. (email: zyan@xidian.edu.cn, ysyangxd@stu.xidian.edu.cn)}%
\thanks{C. Masouros is with the Department of Electronic and Electrical
Engineering, University College London, London WC1E 7JE, U.K. (e-mail:
chris.masouros@ieee.org).}
\thanks{Z. Han is with the Department of Electrical and Computer Engineering, University of Houston, Houston, TX 77004 USA, and also with the Department of Computer Science and Engineering, Kyung Hee University, Seoul 446-701, South Korea (email: hanzhu22@gmail.com)}
 }

\markboth{Journal of \LaTeX\ Class Files}%
{Shell \MakeLowercase{\textit{et al.}}: A Sample Article Using IEEEtran.cls for IEEE Journals}
\maketitle

\begin{abstract}

Integrated sensing and communication (ISAC) systems potentially encounter significant performance degradation in densely obstructed urban and non-line-of-sight scenarios, thus limiting their effectiveness in practical deployments. To deal with these challenges, this paper proposes a backscatter device (BD)-assisted ISAC system, which leverages passive BDs naturally distributed in underlying environments for performance enhancement. These ambient devices can enhance sensing accuracy and communication reliability by providing additional reflective signal paths.
In this system, we define the Pareto boundary characterizing the trade-off between sensing mutual information (SMI) and communication rates to provide fundamental insights for its design. To derive the boundary, we formulate a performance optimization problem within an orthogonal frequency division multiplexing (OFDM) framework, by jointly optimizing time-frequency resource element (RE) allocation, transmit power management, and BD modulation decisions. To tackle the non-convexity of the problem, we decompose it into three subproblems, solved iteratively through a block coordinate descent (BCD) algorithm. Specifically, the RE subproblem is addressed using the successive convex approximation (SCA) method, the power subproblem is solved using an augmented Lagrangian combined water-filling method, and the BD modulation subproblem is tackled using semidefinite relaxation (SDR) methods. Additionally, we demonstrate the generality of the proposed system by showing its adaptability to bistatic ISAC scenarios and MIMO settings.
Finally, extensive simulation results validate the effectiveness of the proposed system and its superior performance compared to existing state-of-the-art ISAC schemes.

\end{abstract}


\section{Introduction}

Integrated sensing and communication (ISAC) has emerged as a pivotal technology for the sixth-generation (6G) wireless networks and beyond \cite{wei2023integrated,zhu2024enabling}, driven by its capability to simultaneously provide communication and radar sensing functionalities. In typical ISAC systems, base stations (BSs) utilize shared wireless resources to simultaneously detect remote targets and communicate with mobile users, significantly reducing hardware complexity while enhancing spectral and energy efficiency  \cite{9705498}. Consequently, ISAC is well-positioned to support various demanding applications in 6G networks, such as autonomous driving, virtual reality (VR), and the low altitude economy ecosystem \cite{10216343,9737357}. Nevertheless, practical ISAC deployments still encounter considerable challenges, such as high energy consumption, limited performance, and insufficient robustness in energy-constrained, urban, non-line-of-sight scenarios, which require further research.

\begin{figure}[tp]
    \centering
    \includegraphics[width=0.95\linewidth]{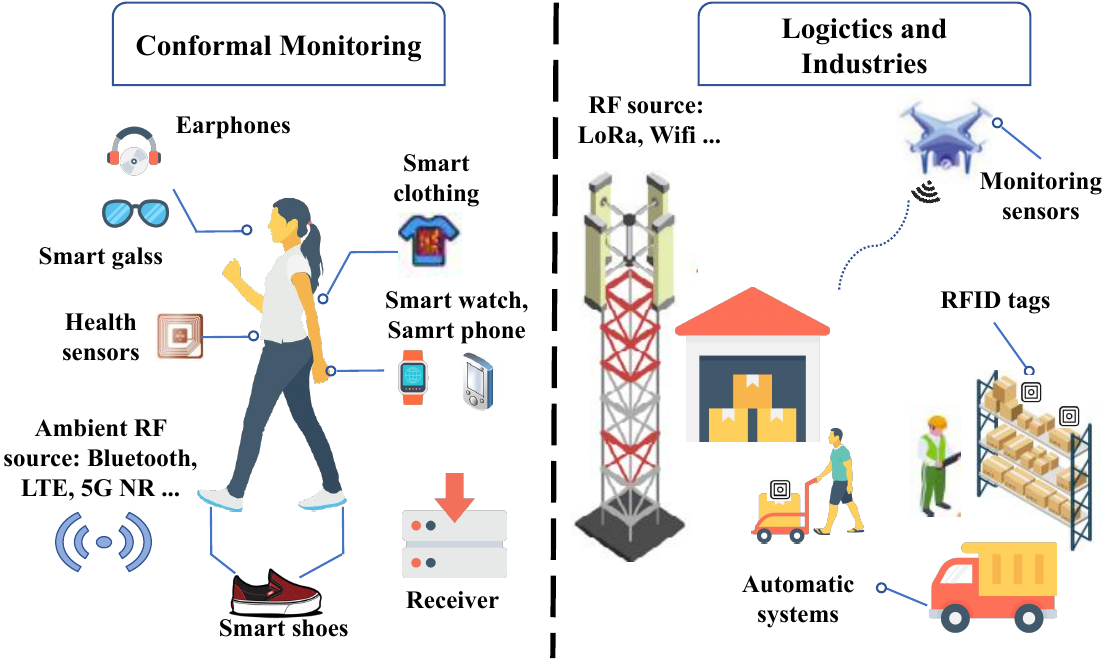}
    \vspace{-0.1cm}
    \caption{Applications of backscatter devices in the real world.}
    \label{BDs}
    \vspace{-0.7cm}
\end{figure}
To enhance ISAC performance while maintaining low cost and energy consumption, economical auxiliary devices such as reconfigurable intelligent surfaces (RISs) and backscatter devices (BDs) have been proposed to improve the received signal strength in ISAC systems.
RISs, composed of programmable metasurface arrays, intelligently shape electromagnetic waves, significantly improving signal quality, coverage, and interference mitigation, thus advancing both sensing and communication performance \cite{liu2023snr,10416996,10496515,10857301}. 
However, despite some RISs being passive or capable of energy harvesting \cite{meng2023ris}, their practical deployment is hindered by limited installation flexibility and high cost.
In contrast, BDs are passive units that modulate on incident ambient signals without generating carrier waves \cite{7948789}, thereby making them inherently energy-efficient and ideal for low-power applications. Unlike RISs, most BDs require no dedicated external power sources. Their compact and simple design makes them extremely cost-effective, with unit low costs of 7 to 15 cents (USD) \cite{8368232}, and small size within a few square centimeters \cite{abdulghafor2021recent}. As shown in Fig. \ref{BDs}, BDs have therefore been broadly deployed in many real-world applications such as conformal monitoring \cite{10130082}, logistics \cite{ahmed2024noma}, and industries \cite{10556582}. Motivated by these advantages, leveraging ambient BDs in ISAC systems offers a compelling approach to enhancing performance across various applications.

\vspace{-0.3cm}
\subsection{Related works}

Recent studies have increasingly focused on using BDs for ultra-low-power and low hardware complexity localization and communication functionalities. For instance, Ren et al. \cite{ren2023toward} leveraged ambient BDs to reduce hardware complexity and improve the accuracy for localization tasks in ISAC systems. They comprehensively characterized complex BD signal propagation environments and demonstrated that incorporating BDs for signal relaying can enhance localization accuracy and mitigate multipath interference.
Xu et al. \cite{10566596} explored joint localization and signal detection in an orthogonal frequency division multiplexing (OFDM)-based ISAC system, where ambient BDs function simultaneously as sensed targets for the BS and communication terminals for users. Given the inherently weak backscatter signals from BDs, the authors designed a specialized link-layer protocol to effectively mitigate environmental interference. Furthermore, Zhang et al. \cite{zhang2024net} proposed a practical ISAC system where BDs serve as the localized target and communication relay for users.
They designed a frame structure that effectively mitigates interference and theoretically analyzed the BD’s location detection probability and symbol error rate at the user, showing the advantages of using BDs for enhancing localization accuracy and communication reliability.

In addition, Galappaththige et al. \cite{10221890} introduced the integrated sensing and backscatter communication (ISBC) concept, where BDs simultaneously serve as the sensed targets while reflecting data signals to users for communication enhancement. They focused on fundamental sensing and communication metrics, i.e., sensing and communication rates, which are optimized by proposing an effective power allocation strategy.
Furthermore, Luo et al. \cite{luo2024isac} and Zhao et al. \cite{zhao2024b} utilized multiple-input multiple-output (MIMO) and beamforming techniques to further improve both sensing and communication capabilities. By optimizing the beamforming vectors at the transmitter, these studies have achieved substantial performance enhancements in monostatic and bistatic configurations, respectively.
\vspace{-0.3cm}
\subsection{Motivations and Contributions}
Despite the advancements presented by these schemes, several research gaps remain:

\begin{itemize}
\item First, the potential of leveraging BDs for performance
enhancement in ISAC systems is not sufficiently researched. The existing literature deliberately deploys BDs as sensing targets or communication terminals. However, the practical scenario of utilizing off-the-shelf ambient BDs distributed in the environment, as signal relays to simultaneously enhance both sensing and communication performance, has been largely overlooked. These ambient BDs naturally reflect sensing echoes and create multipath propagation \cite{9778563}, providing an opportunity to improve system performance without the need for delicate device deployment or additional power supplies.

\item Second, the information-theoretic Pareto boundary characterizing the optimal sensing-communication trade-off remains unexplored in BD-based ISAC systems. Prior studies mainly optimize sensing or communication separately, neglecting their inherent trade-offs. Additionally, most research focuses on localization accuracy rather than sensing mutual information (SMI), a fundamental metric that quantifies the maximum achievable sensing performance from an information-theoretic perspective.

\item Third, existing studies lack effective algorithms for jointly optimizing shared resources in BD-based ISAC systems under an OFDM configuration. These studies typically address isolated resource optimization problems, lacking comprehensive consideration of the joint optimization of multiple system resources, including time-frequency resource element allocation, power allocation, and discrete BD modulations. Consequently, the achievable maximum performance region for sensing and communication remains largely unexplored.

\end{itemize}

To bridge these research gaps, this paper proposes a practical BD-assisted ISAC system that leverages ambient BDs to simultaneously enhance sensing and communication performance. In this system, we first formulate a joint resource optimization problem and characterize the sensing-communication trade-off using Pareto boundaries, considering various optimal variables, including time-frequency resource elements allocation, power allocation, and discrete BD modulation decisions. To solve this complex and nonconvex problem, we propose an efficient optimization algorithm combining a block coordinate descent (BCD) framework and various nonconvex relaxation methods. Based on the Pareto front obtained, we investigate the impact of various system parameters on sensing and communication performance. The primary contributions of this paper are summarized as follows:

\begin{itemize}
\item \textbf{Proposal of a Practical BD-assisted ISAC System.} We introduce a practical system utilizing off-the-shelf BDs naturally distributed in real-world environments. By leveraging BD-reflected sensing echoes and multipath signals, the system enhances both the sensing and communication performance. These ambient BDs commonly function as wearable sensors, environmental monitors, or logistics tracking tags near typical ISAC entities, eliminating the need for dedicated deployment, active power supplies, and precise control, thereby simplifying implementation and reducing complexity.

\item \textbf{Characterization of the Pareto Boundary.} To characterize the optimal radar-communication trade-off, we formulate a Pareto boundary searching problem \cite{gao2023cooperative,10540149}, using the SMI and communication rate as fundamental metrics. We rigorously establish the compactness and normality of the achievable performance region. To determine the Pareto boundary, we design two constrained optimization problems: sensing-centric optimization and communication-centric optimization, in which time-frequency resource element allocation, power allocation, and BD modulation decision are jointly considered.

\item \textbf{Joint Multi-dimensional Resource Allocation.}
To effectively solve the Pareto optimization problem, we propose a BCD-based algorithm that decomposes it into three subproblems: resource element allocation, power allocation, and BD modulation decisions. Specifically, we apply successive convex approximation (SCA) for resource allocation, augmented Lagrangian combined with water-filling for power allocation, and semidefinite relaxation (SDR) for BD modulation, thus obtaining a near-optimal Pareto boundary. Additionally, we discuss the algorithm's adaptability to bistatic ISAC systems and MIMO settings.


\end{itemize}

Extensive simulations demonstrate that the proposed BD-assisted ISAC system significantly outperforms state-of-the-art and non-optimized schemes across various scenarios. Numerical results reveal approximately 15\% improvements in both the SMI and the communication rate compared to conventional ISAC systems without leveraging BDs, thereby validating the effectiveness and superiority of the proposed algorithm. Additionally, the proposed system achieves cost reductions ranging from 50\% to 80\% relative to a RIS-assisted benchmark, while incurring less than a 5\% performance penalty, underscoring its substantial potential for practical deployment.

\textit{Remainder}: Section~\ref{System model} describes the system model and performance metrics. Section~\ref{Problem Formulation} formulates the Pareto boundary to characterize sensing-communication trade-offs and defines corresponding optimization problems. Section~\ref{Pareto Boundary} proposes an optimization algorithm for solving these problems. Simulation results are provided in Section~\ref{performance}, and conclusions are drawn in Section~\ref{Conclusion}.

\begin{figure}[tp]
    \centering
    \includegraphics[width=0.95\linewidth]{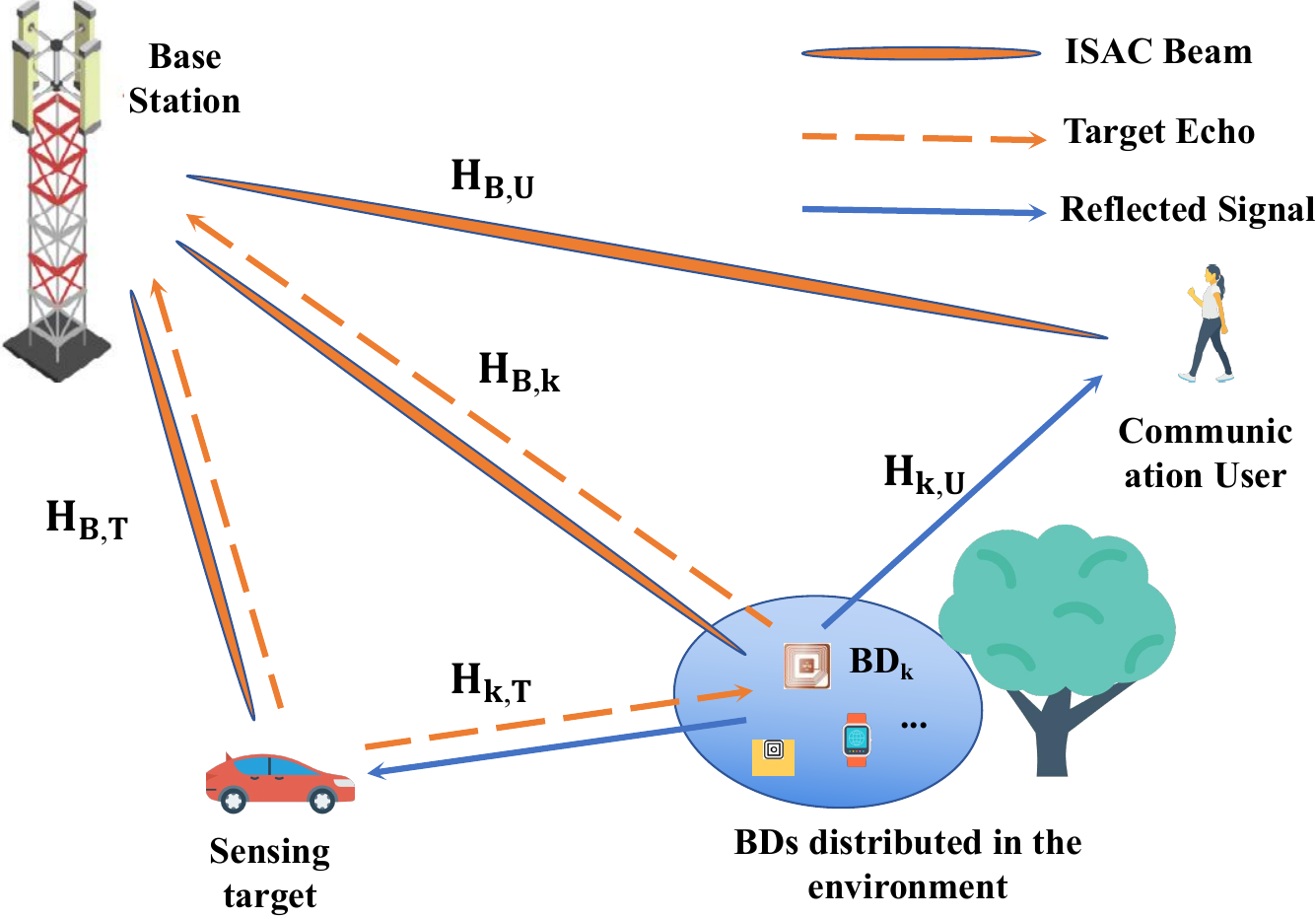}
    \caption{System model of a BD-assisted ISAC system.}
    \label{syse}
    \vspace{-0.5cm}
\end{figure}

\vspace{-0.2cm}
\section{System Model} \label{System model}
As shown in Fig. \ref{syse}, we consider a BD-assisted monostatic ISAC system comprising a full-duplex base station (BS) transmitting OFDM samples to communicate with user equipment (UE) while simultaneously sensing a target using reflected OFDM echoes. Multiple BDs, closely distributed near entities such as the BS, the UE, and the target, reflect incident signals, including the OFDM transmissions and the target echoes. The detailed transmit signal model, backscatter model, and received signal model are presented in the following subsections.

\vspace{-0.2cm}
\subsection{Transmit Signal Model}

An OFDM waveform is adopted for both radar sensing and data transmission. Specifically, a processing interval comprises a frame of $M$ consecutive OFDM symbols, each containing $N$ subcarriers. Each resource element (RE) for the $n$-th subcarrier of the $m$-th OFDM symbol is denoted by $(m,n)$ in the frequency-time domain. We define the set of OFDM symbols as $\mathcal{M} = \{0, \dots, M - 1\}$ and the subcarrier indices as $\mathcal{N} = \{0, \dots, N - 1\}$. For each symbol $m$, the subcarriers are partitioned into radar subcarriers $\mathcal{N}_{r,m}$ and communication subcarriers $\mathcal{N}_{c,m}$. 
Radar subcarriers transmit deterministic, constant-envelope pilot symbols that guarantee stable sensing performance and could also be used for channel estimation at users, whereas communication subcarriers convey random data payload for information transmission. As this paper primarily addresses the optimization of time-frequency resources, we simplify the system to a single-input single-output (SISO) configuration \cite{zhang2024cross}. The extension of the proposed approach to MIMO systems is discussed further in Section \ref{discussion}.

Based on the above definition, the transmitted signal of the BS is expressed as:

\vspace{-0.5cm}
\begin{align}
x(t) = \sum_{m \in \mathcal{M}} &\frac{p(t - t_m)}{N} 
\Bigg( \underbrace{\sum_{n \in \mathcal{N}_{r,m}} X_{r,n,m} e^{j 2 \pi f_n (t - t_m)}}_{\text{Radar subcarriers}} \nonumber \\
&+ \underbrace{\sum_{n \in \mathcal{N}_{c,m}} X_{c,n,m} e^{j 2 \pi f_n (t - t_m)}}_{\text{Communticaion subcarriers}} \Bigg),
\end{align}
where $p(t)$ is the pulse-shaping function, $X_{r,n,m}$ and $X_{c,n,m}$ are the transmitted  radar and communication signal respectively, $f_n$ is the frequency of the $n$-th subcarrier, and $t_m$ is the starting instant of the $m$-th OFDM symbol. Typically, the communication symbols are complex Gaussian distributed, $ X_{c,n,m} \sim \mathcal{CN}(0, \sigma_{\text{c}}^2)$ \cite{tse2005fundamentals}, while radar subcarrier symbols $X_{r,n,m}$ remain constant \cite{wei20225g}. A cyclic prefix (CP) of duration $T_G$ is added to each OFDM symbol to prevent inter-symbol interference, making the total OFDM symbol duration $T_o = T + T_G$.

The instantaneous powers allocated to radar and communication subcarriers are defined respectively as:
\begin{equation}
    P_r = \sum_{m \in \mathcal{M}} \sum_{n \in \mathcal{N}_{r,m}} P_{n,m}, \quad
    P_c = \sum_{m \in \mathcal{M}} \sum_{n \in \mathcal{N}_{c,m}} P_{n,m}, 
\end{equation}
where $P_{n,m} = |X_{n,m}|^2$ representing the transmit power at RE $(n,m)$ and $P_b = P_r + P_c$ denoting total transmit power of the BS.

\vspace{-0.4cm}
\subsection{Backscatter Model} \label{Backscatter Model}
We consider a practical scenario where off-the-shelf BDs are placed at fixed locations, each characterized by deterministic backscatter patterns \cite{10566596}. $K$ BDs perform phase-shift keying (PSK) backscatter, reflecting signals by altering the phase via load impedance switching on each OFDM sample \cite{qian2018iot}. The backscatter modulation symbol of the $k$-th BD for the $m$-th OFDM symbol is represented as $b_{k,m} = \alpha_k e^{j\phi_{k,m}}$, where:
\begin{equation}
\phi_{k,m} = \begin{cases}
0, & \text{for bit } ``1", \\
    \pi, & \text{for bit } ``0",
\end{cases}
\end{equation}
with attenuation coefficient $\alpha_k \in [0,1]$. Thus, the backscatter modulation of the $k$-th BD over $M$ symbols is expressed as:
\begin{equation}
\mathbf{b_k}(t) = \sum_{m=0}^{M-1} b_{k,m} \mathrm{rect}\left(\frac{t - mT_{O}}{T_{O}}\right).
\end{equation}

It is worth noting that this paper considers two BD-modulation scenarios: \textbf{\textit{Scenario 1}} (fixed modulation)~\cite{7948789}, where each BD employs a predetermined length-$M$ sequence suitable for low-cost passive BDs or slow-fading channels; and \textbf{\textit{Scenario 2}} (dynamic modulation)~\cite{7419631}, where the BS or receivers adjust BD reflection patterns in real-time. Unlike RISs, BDs operate asynchronously, reflecting signals at discrete, unsynchronized time instants governed by each device's local control or energy-harvesting schedules, rather than centralized phase alignment.

\vspace{-0.2cm}
\subsection{Receive Signal Model} 
\subsubsection{Communication Receiver}

In the system, the user can receive the communication signals directly from the BS and the signal backscattered by BDs. 
Thus, at the user end, after the CP removal, the received signal at the communication RX is given by
\vspace{-0.1cm}
\begin{equation}
   Y_{c}(m,n) = (H_{b,u}+ \sum_{k=1}^K H_{b,k} b_{k,m} H_{k,u} ) X(m,n) + V_{c}(m,n), \label{comm receive}
\end{equation}
where $V_{c}(m,n) \sim \mathcal{CN}(0,\sigma_c^2)$ is additive white Gaussian noise (AWGN). Channels $H_{b,u}$, $H_{b,k}$, and $H_{k,u}$ represent propagation from the BS to the UE, the BS to the BD, and the BD to the UE, respectively. In $(n,m)$, the frequency domain response for between device $\alpha$ and $\beta$ can be modeled as:
\vspace{-0.1cm}
\begin{equation}
H_{\alpha, \beta}(n,m) = \sum_{l=1}^{S_l} A_{l} e^{-j 2\pi f_n \tau_{l}} e^{j 2\pi f_{D,l} t_m}, \label{channel}
\end{equation}
where $(\alpha, \beta)\in \{b,u,k\}$, $S_l$ and $A_{l}$ are the number of propagation paths in the environment and the corresponding attenuation constant for the path between each scatterer, $\tau_{l}$ and $f_{D,l}$ are the delay and Doppler shift of each path.


Based on \eqref{comm receive}, the data communication rate is then given as
\vspace{-0.1cm}
{ \small
\begin{align}
C_d = \frac{\Delta f}{T_{{o}}} & \sum_{m \in \mathcal{M}} \sum_{n \in \mathcal{N}_{c,m}} \log_2 \Bigg( 1 + \frac{\Big| H_{b,u}+ \sum_{k=1}^K b_{k,m} H_{b,k} H_{k,u} \Big|^2} {\sigma_c^2}
\nonumber \\ &   \times \mathbb{E} \{|X_{c}(m,n)|^2\} \Bigg),
\end{align}
}
where $\Delta f$ denotes the frequency spacing.

\subsubsection{Sensing Receiver} 
The BS receives baseband signals from reflections, including target echoes and BD reflections. We assume that the BS is full-duplex \cite{8805161}, so that self-interference can be canceled in the BS.
Thus, the baseband signal received by the sensing receiver is given by
\vspace{-0.1cm}
\begin{align}
 Y_{r}(m,n) = \alpha_t (&H_{s,t}   + \sum_{k=1}^K b_{k,m} H_{s,k} H_{k,t})(H_{s,t}  \nonumber \\ &+ \sum_{k=1}^K b_{k,m} H_{s,k} H_{k,t})^H X(m,n)   + V_{r}(m,n),
 \label{radar receiver}
\end{align}
where $\alpha_t \sim \mathcal{CN}(0, \sigma_t^2)$ represents the radar cross section (RCS) of the sensing target. $H_{s,t}$ and $H_{k,t}$ denote the channel between the BS to the target and the channel between $k$-th BD and the target, respectively, which also follows \eqref{channel}.  $V_{r}(m,n)  \sim \mathcal{CN}(0,\sigma_r^2)$ denote the noise in frequency domain.  Since the communication subcarriers are also known in the BS, their extra sensing gains are also considered \cite{liyanaarachchi2021optimized}.

Since the SMI relates to various sensing metrics, such as minimum mean square error (MMSE), angle of arrival (AoA), and distance, and can be unified with communication metrics to measure ISAC performance~\cite{10540149,bicua2019multicarrier}, we adopt SMI to evaluate the sensing performance of the proposed system. Let $G =(H_{s,t}   
+ \sum_{k=1}^K b_{k,m} H_{s,k} H_{k,t}) $, the SMI can be calculated as 
\vspace{-0.2cm}
{ \small
\begin{equation}
I_r = \frac{\Delta f}{T_{{o}}}  \sum_{m \in \mathcal{M}} \sum_{n \in \mathcal{N}_m} 
\log_2 \Bigg( 1 +  \frac{| \alpha_t|^2| GG^H|^2}
{\sigma_r^2}\times \mathbb{E} \{|X(m,n)|^2\} \Bigg),
\label{eq:Ir_again}
\end{equation}}
where the additional sensing gain contributed by communication subcarriers is also considered \cite{liyanaarachchi2021optimized}.

\vspace{-0.3cm}
\section{Pareto Optimization Boundary for the Proposed System} \label{Problem Formulation}

In this section, we first formulate the optimization problem for the BD-assisted ISAC system, considering various practical constraints. Then, we define the achievable performance region between sensing and communication using the Pareto boundary and present approaches to obtaining this optimal boundary.

\vspace{-0.3cm}
\subsection{Pareto Optimization Boundary Formulation}

\textbf{\textit{Remark 1} (Pareto Optimization \cite{gao2023cooperative,10540149}):} In multi-objective optimization, traditional optimality is typically replaced by Pareto optimality, defined as achieving the best possible performance in one objective without degrading the performance in another.

In this paper, we focus on the unified downlink communication and radar sensing based on the proposed system, where the BS cooperatively designs the transmission policies, including the RE allocation, power allocation, and BD modulation decisions to optimize dual-functional performances. 
Mathematically, the corresponding optimization problem can be formulated as
\begin{subequations}
\begin{align}
\max_{\mathcal{N}_{n,m}, P_{n,m},\Phi_m} \quad &   f  (I_c, I_s) \\ 
\textit{s.t.} & \quad I_r \geq 0, C_d \geq 0, \label{eq:SC_limit} \\
& \quad P_c + P_r \leq P_t, \label{eq:power_limit} \\
& \quad 0\leq P_{n,m} \leq P_{\text{max}}, \label{eq:signal_power} \\
& \quad \mathcal{N}_{n,m} \in \{0,1\},  \label{eq:binary_alloc}  \\
& \quad  \phi_{k,m} \in \{0, \pi\}, \quad \forall k \in \mathcal{K}, \label{eq:modulation_alloc}
\end{align}
\end{subequations}
where $f(I_r, C_d)$ denotes the search for Pareto-maximal solutions under the component-wise partial order, and $\Phi_m = \{\phi_{1,m}, \dots, \phi_{K,m}\}$ represents the modulation decisions of $K$ BDs for the $m$-th OFDM sample. Constraint~\eqref{eq:SC_limit} ensures feasible sensing and communication performance. Constraint~\eqref{eq:power_limit} restricts total transmit power ($P_c + P_r$) within the maximum budget $P_t$. Constraint~\eqref{eq:signal_power} limits the power at each resource element (RE) to $P_{\text{max}}$, constraint~\eqref{eq:binary_alloc} enforces binary allocation of REs for radar or communication usage, and constraint~\eqref{eq:modulation_alloc} restricts the concrete BD phase modulation.


Due to shared time-frequency REs, transmit power, and BD modulation decisions, there exists an inherent trade-off between achievable SMI and communication rate. Characterizing and achieving this optimal trade-off is essential to guide practical BD-assisted ISAC system design. Next, we define the performance region and discuss methods to determine the optimal sensing-communication trade-off.


\vspace{-0.5cm}
\subsection{Pareto Boundary of the proposed system}
In this subsection, we provide the necessary definitions and lemmas for the Pareto boundary.

\textit{Definition 1: Achievable Performance Region.} For the joint communication and sensing system, the mutual performance region under the defined constraint is defined as follows:
\begin{align}
\mathcal{R} = \{(I_s, C_d) : C_d \leq C_d^{\max}, I_s \leq I_s^{\max},  \eqref{eq:SC_limit}-\eqref{eq:binary_alloc} \}, \label{region}
\end{align}
where $I_s$ and $C_d$ denote the achievable SMI and communication rate, respectively. The two-dimensional region $\mathcal{R}$ characterizes all simultaneously achievable sensing and communication performance pairs under feasible transmission strategies. The shape of $\mathcal{R}$ depends on factors such as channel conditions, power constraints, and system parameters. Next, we present several properties of $\mathcal{R}$ to characterize the Pareto boundary.

\textit{Definition 2 (compact and normal set \cite{chen2021joint}):} A set $G \subseteq \mathbb{R}_+^n$ is called compact if it is closed and bounded. A set $G \subseteq \mathbb{R}_+^n$ is called a normal region if for any two points $x \in G$ and $x' \in \mathbb{R}_+^n$, if $x' \leq x$ (element-wise), then $x' \in G$.

\begin{lemma}
The achievable performance region $\mathcal{R}$ of the proposed system defined in (\ref{region}) is compact and normal. 
\end{lemma}

\begin{proof}
The proof is given in Appendix \ref{appendice1}.
\vspace{-0.2cm}
\end{proof}

These properties are crucial to guarantee the existence of an optimal Pareto solution and ensure that the optimal performance can always be achieved on the Pareto boundary.

\textit{Definition 3: Pareto Boundary:} The Pareto boundary of the performance region is defined as the set of points $(I_s, C_d)$ such that no other achievable performance pair $(I'_s, C_d') \in \mathbb{R}_+^n$ satisfies both $I'_s \geq I_s$ and $C_d' \geq C_d$ simultaneously. Mathematically, the Pareto boundary is given by:
\begin{align}
\mathcal{K} = \Big\{(I_s, C_d) : & \nexists (I'_s, C_d') \in \mathbb{R}_+^n, I'_s \geq I_s, C_d' \geq I_s, \nonumber \\
& \eqref{eq:SC_limit}-\eqref{eq:binary_alloc} \Big\}.
\end{align}

Next, we formulate the problem to search the Pareto boundary.

\vspace{-0.5cm}
\subsection{Pareto Boundary Searching}


In our problem, due to the nonconvexity of the objective function and the feasible set, $\mathcal{R}$ is generally non-convex. To find the Pareto optimal points, i.e., the whole Pareto boundary, of the achievable performance region $\mathcal{R}$ in the proposed system, we consider the following two optimization problems:

\subsubsection{Sensing Maximization under Communication Constraint (Sensing-centric optimization problem)}
\begin{subequations}\label{P1}
\begin{align}
\text{(P1)}   \max_{\mathcal{N}_{r,m}, P_{n,m},\Phi_m}  &I_r(\mathcal{N}_{r,m}, P_{n,m},\Phi_m) \label{eq:obj} \\
\textit{s.t.} & \quad C_d(\mathcal{N}_{r,m}, P_{n,m},\Phi_m) \geq \Gamma_c, \label{eq:comm_rate} \\
& \quad \eqref{eq:SC_limit}-\eqref{eq:binary_alloc},
\end{align}
\end{subequations}
where $\Gamma_c$ represents the minimum required communication rate.

\subsubsection{Communication Maximization under Sensing Constraint (Communication-centric optimization problem)}
\begin{subequations}\label{P2}
\begin{align}
\text{(P2)} \quad \max_{\mathcal{N}_{r,m}, P_{n,m},\Phi_m} \quad & C_d(\mathcal{N}_{r,m}, P_{n,m},\Phi_m) \\ \label{14a}
\textit{s.t.} & \quad I_s(\mathcal{N}_{r,m}, P_{n,m},\Phi_m) \geq \Gamma_s, \\\label{14b}
&\quad \eqref{eq:SC_limit}-\eqref{eq:binary_alloc},
\end{align}
\end{subequations}
where $\Gamma_s$ represents the minimum required SMI.


The complete Pareto boundary can be obtained by the constrained method, which is guaranteed by the following lemma.

\begin{lemma}
The complete Pareto boundary of the achievable performance region $\mathcal{R}$ defined in \eqref{region} can be characterized by the solutions to problem \eqref{P1} with different communication constraints $\Gamma_c$, or by the solutions to problem \eqref{P2} with different sensing constraints $\Gamma_s$.
\end{lemma}

\begin{proof}
The proof is given in Appendix \ref{appendice2}.
\end{proof}


\textbf{\textit{Remark 2} (Equivalence of Problem):} It can be observed that \eqref{P1} and \eqref{P2} have common optimized variables and structure in the SMI and communication rate functions. Also, they can be transformed into each other by exchanging the object function and the performance constraint. Thus, a similar algorithm can solve them simultaneously. In the next section, we take \eqref{P1} as an example and propose an efficient algorithm to approximately solve it in the proposed system, while it is worth noting that both \eqref{P1} and \eqref{P2} are implemented in Section \ref{performance} for performance evaluation.

\vspace{-0.3cm}
\section{Approximately Pareto Optimal Resource Allocation For Joint Sensing and Communication} \label{Pareto Boundary}

This section proposes a near-optimal solution to characterize the Pareto boundary. The original optimization problem is decomposed into three manageable subproblems, each solved separately via suitable nonconvex relaxation methods and iteratively updated within a BCD framework.

\vspace{-0.3cm}
\subsection{Proble Decomposition}

Problem (P1) is inherently non-convex due to coupling between the RE allocation, power allocation variables, and binary constraint, making it NP-hard to solve optimally~\cite{gao2023cooperative}. To overcome this, we decompose (P1) into three subproblems: RE allocation optimization, power allocation optimization, and BD modulation decision optimization.

\textbf{\textit{RE Allocation Optimization:}} When power allocation and BD modulation decision are fixed, (P1) reduces to:

\vspace{-0.3cm}
\begin{subequations}
\begin{align}
\text{(P3)} \quad \max_{\mathcal{N}_{r,m}} \quad & I_r(\mathcal{N}_{r,m}) \\
\textit{s.t.} & \quad \eqref{eq:comm_rate}, \ \eqref{eq:binary_alloc}.
\label{eq:subcarrier_allocation_original}
\end{align}
\end{subequations}

\textbf{\textit{Power Allocation Optimization:}} With fixed RE allocation and BD modulation decision, problem (P1) reduces to:
\begin{subequations}
\begin{align}
\text{(P4)} \quad \max_{P_{n,m}} \quad&I_r(P_{n,m}) \label{P1.1a}\\
\textit{s.t.} & \quad \eqref{eq:comm_rate}, \ \eqref{eq:power_limit}, \ \eqref{eq:signal_power}.
\end{align}
\end{subequations}

\textbf{\textit{BD Modulation Decision Optimization:}} With fixed RE allocation and power allocation, problem (P1) reduces to:
\begin{subequations}
\begin{align}
\text{(P5)} \quad \max_{\Phi_m} \quad&I_r(\Phi_m) \label{P1.1a}\\
\textit{s.t.} & \quad \eqref{eq:comm_rate}, \ \eqref{eq:modulation_alloc}.
\end{align}
\end{subequations}

Following the BCD framework, problem (P1) can be approximately solved by iteratively updating the solutions to the three subproblems until convergence~\cite{tseng2001convergence}. In the next subsection, we propose an algorithm to iteratively solve these subproblems.

\textbf{\textit{Remark 3} (Two Cases for BDs Modulation Optimization):} As mentioned in subsection \ref{Backscatter Model}, two cases for BD modulations are considered. In the case that BDs cannot be timely controlled (\textbf{\textit{Scenario 1}}), the subproblem (P5) is not considered in the optimization procedure. In the case of BDs that can be timely controlled (\textbf{\textit{Scenario 2}}), all three above subproblems are considered for optimization. 
\vspace{-0.2cm}
\subsection{Proposed Optimization Algorithm}


\subsubsection{Resource Elements Optimization Allocation}

 To solve the RE allocation problem, we first introduce an auxiliary indicator for radar and communication allocation:
\begin{equation}
\mathbb{I}_{\mathcal{N}_{r,m}}(n) =
\begin{cases}
1, & n \in \mathcal{N}_{r,m}, \\
0, & n \in \mathcal{N}_{c,m}.
\end{cases}
\label{eq:indicator_function}
\end{equation}

Formally, the optimization problem (P1) now becomes:
\begin{subequations}
\begin{align}
\text{(P4)} \quad \max_{\mathbb{I}_{\mathcal{N}_{r,m}}} \quad & I_r(\mathbb{I}_{\mathcal{N}_{r,m}}) \\
\textit{subject to} & \quad \eqref{eq:comm_rate}, \ \eqref{eq:binary_alloc}.
\label{eq:subcarrier_allocation_original}
\end{align}
\end{subequations}

We can find the subproblem (P4) is a mixed-integer program, which is still nonconvex. To tackle the combinatorial nature of the binary constraint, we relax \(\mathbb{I}_{\mathcal{N}_{r,m}}(n)\) to \([0,1]\) using an \(\ell_0\)-norm approximation. In particular, we have
\begin{equation}
\mathbb{I}_{\mathcal{N}_{r,m}}(n) = \|\Xi_{m,n}^H\mathbf{z}\|_0,
\label{eq:l0_norm}
\end{equation}
and approximate the \(\ell_0\)-norm by the smooth function
\begin{equation}
\|w\|_0 \approx g(w;\delta) = \frac{\ln\Bigl(1+\frac{w}{\delta}\Bigr)}{\ln\Bigl(1+\frac{1}{\delta}\Bigr)}, \quad \forall\,w \ge 0,
\label{eq:smooth_approx}
\end{equation}
with \(\delta > 0\) a small constant. Thus, we relax the binary constraint \eqref{eq:binary_alloc} as
\begin{equation}
g(\mathbb{I}_{\mathcal{N}_{r,m}}(n);\delta) - \mathbb{I}_{\mathcal{N}_{r,m}}(n) \le 0, \quad \forall\,(m,n).
\label{eq:relaxed_constraint}
\end{equation}

Since \(g(\mathbb{I}_{\mathcal{N}_{r,m}}(n);\delta)\) is still nonconvex and the smooth function $g()$ is upper concave, we linearize it \(\mathbb{I}_{\mathcal{N}_{r,m}} (n)\) via a first-order Taylor expansion:
\begin{align}
g(\mathbb{I}_{\mathcal{N}_{r,m}}(n)&;\delta) \le g\Bigl(\mathbb{I}_{\mathcal{N}_{r,m}} (n);\delta\Bigr) \nonumber
\\
&+ \nabla g\Bigl(\mathbb{I}_{\mathcal{N}_{r,m}} (n);\delta\Bigr)
\Bigl(\mathbb{I}_{\mathcal{N}_{r,m}}(n) - \mathbb{I}_{\mathcal{N}_{r,m}} (n)\Bigr),
\label{eq:taylor_expansion}
\end{align}

where
\vspace{-0.3cm}
\begin{align}
\gamma_{m,n}  &\triangleq \frac{1}{\ln\Bigl(1+\frac{1}{\delta}\Bigr)}\,\frac{1}{\delta+\mathbb{I}_{\mathcal{N}_{r,m}} (n)}, \label{eq:gamma_def} \\
\kappa_{m,n}  &\triangleq g\Bigl(\mathbb{I}_{\mathcal{N}_{r,m}} (n);\delta\Bigr) - \gamma_{m,n}^{(t)}\,\mathbb{I}_{\mathcal{N}_{r,m}} (n). \label{eq:kappa_def}
\end{align}

Then, the surrogate constraint constraint in \eqref{eq:relaxed_constraint} becomes
\begin{equation}
-\gamma_{m,n} \,\mathbb{I}_{\mathcal{N}_{r,m}}(n) + \kappa_{m,n}  - \mathbb{I}_{\mathcal{N}_{r,m}}(n)\le 0, \quad \forall\,(m,n).
\label{eq:surrogate_constraint}
\end{equation}

The subcarrier allocation subproblem now becomes
\vspace{-0.1cm}
\begin{subequations}
\begin{align}
\text{(P5)} \quad \max_{\mathcal{N}_{r,m}} \quad  &I_r(\mathcal{N}_{r,m}) \\
\textit{s.t.} & \quad \eqref{eq:comm_rate}, \nonumber\\[1mm]
& -\gamma_{m,n} \,\mathbb{I}_{\mathcal{N}_{r,m}}(n) + \kappa_{m,n}  - \mathbb{I}_{\mathcal{N}_{r,m}}(n)\le 0,  \\
& \quad  0 \le \mathbb{I}_{\mathcal{N}_{r,m}}(n) \le 1.
\label{eq:subcarrier_allocation_modified}
\end{align}
\end{subequations}





This problem is convex and can be solved efficiently using convex optimization tools such as CVX. After solving the relaxed problem, we apply a thresholding step to recover the discrete subcarrier allocation:
\begin{equation}
\mathbb{I}_{\mathcal{N}_{r,m}}(n) =
\begin{cases}
1, & \text{if } \mathbb{I}_{\mathcal{N}_{r,m}}(n) \ge 0.5, \\
0, & \text{otherwise}.
\end{cases}
\label{eq:thresholding}
\end{equation}

\subsubsection{Power Optimization Allocation}

Problem (P4) is non-convex because it maximizes a concave SMI term. To obtain a tractable solution, we apply an augmented Lagrangian method to reformulate the constrained problem as follows: 
\begin{align}
\mathcal{L}(P_{m,n}, \lambda, \rho) &=I_r(P_{m,n}) +\nonumber \\
&  \lambda \left(C(P_{m,n}) - \Gamma_c\right) - \frac{\rho}{2}\left(C(P_{m,n}) - \Gamma_c\right)^2,
\label{eq:augmented_lagrangian}
\end{align}
where $\lambda \geq 0$ is the Lagrange multiplier and $\rho > 0$ is a penalty parameter. The multiplier update rule is: 
\vspace{-0.1cm}
\begin{equation}
        \lambda \leftarrow \max\left\{0, \lambda + \rho(C(P_{m,n}) - \Gamma_c)\right\}. 
    \label{lambda}
\end{equation}

The penalty parameter update rule is: 
\begin{equation}
\rho_{k+1} = \rho_k \cdot \max\left(2, \frac{\|\nabla  C\|_k}{\|\nabla  C\|_{k-1}}\right),
    \label{o}
\end{equation}
which ensures that the optimization process always takes place in a locally convex region, thus guaranteeing the convexity of the problem (P3).

Now, the original constrained optimization problem (P2) can be reformulated as follows:
\vspace{0.2cm}
\begin{subequations}
\begin{align}
\text{(P6)} \quad\max_{P_{m,n}} \quad&\mathcal{L}(P_{m,n}, \lambda, \rho) \label{P1.1a}\\
\textit{s.t.} & \quad \eqref{eq:power_limit},  \text{and} \ \eqref{eq:signal_power},
\end{align}
\end{subequations}
where the goal is to maximize the augmented Lagrangian function \( \mathcal{L}(P_{m,n}, \lambda, \rho) \) with respect to the power allocation \( P_{m,n} \). It can be found that the problem (P3) is convex, which can be solved using gradient-based optimization methods, and the  solution must satisfy the following Karush–Kuhn–Tucker (KKT) stationarity conditions:


\vspace{-0.3cm}
\begin{align}
\nabla_{P_{m,n}} &I_r(P_{m,n}^*) + \lambda^* \nabla_{P_{m,n}} C(P_{m,n}^*) \nonumber \\
&- \rho \bigl(C(P_{m,n}^*) - \Gamma_c\bigr) \nabla_{P_{m,n}} C(P_{m,n}^*) = 0,  \label{KKT}
\end{align} 
where $P_{m,n}^*$ denotes the solution in each iteration of the BCD algorithm. 

It can be found that the stationarity condition \eqref{KKT} admits a water-filling-type–type structure where we obtain a closed-form solution for optimal power allocation in a water-filling fashion:

\vspace{-0.3cm}
\begin{equation}
\resizebox{\columnwidth}{!}{$
P^{*}(m,n)=
\begin{cases}
\displaystyle
\left[
\frac{\ln 2 \big/ \Delta f}{-\nu}
\;-\;
\frac{\sigma_r^{2}}{|\alpha_t|^{2}\,|G(m,n)|^{2}}
\right]^{+},
& (m,n)\in\mathcal N_{r,m},\\[10pt]
\displaystyle
\left[
\frac{\ln 2 \big/ \Delta f}{-\nu-\lambda+\rho\!\left(C-\Gamma_c\right)}
\;-\;
\frac{\sigma_c^{2}}{|H_{c}(m,n)|^{2}}
\right]^{+},
& (m,n)\in\mathcal N_{c,m},
\end{cases}
$}
\label{eq:optimal_power_allocation}
\end{equation}
where $H_c$ = $H_{b,u}+\sum_{k=1}^K b_{k,m}H_{b,k}H_{k,u}$, \(\nu\) is the Lagrange multiplier for the total power constraint \eqref{eq:power_limit}), \(\eta(m,n)\) is the communication constraint penalty term given by $\eta(m,n) = \rho\left(C(P_{m,n}) - \Gamma_c\right) \frac{\partial C}{\partial P(m,n)}$, \([\cdot]^+\) indicates projection onto non-negative values. The penalty term adjusts the power allocation to ensure the communication constraint.

The required gradients for the iterative updates are given explicitly by:
\begin{equation}
\frac{\partial I_r}{\partial P(m,n)} = \Delta f\,\frac{1}{\ln 2}\,\frac{|\alpha_tG G ^H|^2}{\sigma_r^2 + |\alpha_tG G ^H|^2\,P(m,n)},  \label{gradient updates1}
\end{equation}
\vspace{-0.2cm}
\begin{equation}
\frac{\partial C}{\partial P(m,n)} = \Delta f\,\frac{1}{\ln 2}\,\frac{\left|H_c\right|^2}{\sigma_c^2 + \left|H_c\right|^2\,P(m,n)},  \label{gradient updates2}
\end{equation} 
where the derivative procedure is given in Appendix \ref{appendice3}.




\subsubsection{BD Modulation Decision Optimization}

Since the phase-shift elements of BDs are chosen from the discrete set $\{0, \pi\}$. Equivalently, if $b_{k,m} = \alpha_k e^{j\phi_{k,m}}$ with $\phi_{k,m} \in \{0,\pi\}$, we can write 
\begin{equation}
b_{k,m} = \alpha_k x_{k,m}, 
\quad
x_{k,m} \in \{\pm 1\}. \label{xtransform}
\end{equation}

Collect the BD modulation variables for a given symbol~$m$ into $\mathbf{x} \;=\; [\,x_{1,m}, \dots, x_{K,m}\,]^\mathrm{T}\;\in\;\{\pm1\}^{K}$. We focus on one symbol index $m$ for clarity, although the derivation extends straightforwardly if each $m$ has its own $\mathbf{x}_m$.

Define the effective channel gains of $m$-th sample in each iteration for communication and sensing as
\begin{align}
H_m(\mathbf{x}) &= H_{b,u} \;+\; \sum_{k=1}^{K} \alpha_k\, x_k\, H_{b,k}\,H_{k,u}, \\
G_m(\mathbf{x}) &= H_{s,t} \;+\; \sum_{k=1}^{K} \alpha_k\, x_k\, H_{s,k}\,H_{k,t}.
\end{align}

Each $|G_m(\mathbf{x})|^2$ and $|H_m(\mathbf{x})|^2$ is a Hermitian quadratic form in $\mathbf{x}$:
\begin{align}
\bigl|G_m(\mathbf{x})\bigr|^2 \;=\; \mathbf{x}^\mathrm{H}\,\mathbf{Q}_r^{(m)}\,\mathbf{x}, \quad
\bigl|H_m(\mathbf{x})\bigr|^2 \;=\; \mathbf{x}^\mathrm{H}\,\mathbf{Q}_c^{(m)}\,\mathbf{x},
\end{align}
where
\begin{align}
\mathbf{Q}_r^{(m)} &= (\mathbf{h}_r^{(m)}+\mathbf{h}_r^{(0)})(\mathbf{h}_r^{(m)}+\mathbf{h}_r^{(0)})^\mathrm{H}, \\[3pt]
\mathbf{Q}_c^{(m)} &= (\mathbf{h}_c^{(m)}+\mathbf{h}_c^{(0)})(\mathbf{h}_c^{(m)}+\mathbf{h}_c^{(0)})^\mathrm{H},
\end{align}
and $\mathbf{h}_r^{(m)} \;=\;
[\,
\alpha_1 H_{s,1}H_{1,t},\,\dots,\,\alpha_K H_{s,K}H_{K,t}
]^\mathrm{T}$.

\begin{algorithm}[!t]
\footnotesize
\SetCommentSty{small}
\LinesNumbered
\caption{BCD Algorithm}
\label{alg:ISAC_Alternating_Optimization_Revised}
\KwIn{$\alpha_t$, $G$, $\sigma_r^2$, $\sigma_c^2$, $P_t$, $\Gamma_c$, $\Delta f$, $\epsilon_1$, $\epsilon_2$, $\epsilon_3$.}
\KwOut{$\mathcal{N}_{r,m}$, $P_{m,n}$, and $\Phi_m$.}

Randomly initialize $\mathcal{N}_{r,m}^{(0)}$, $P_{m,n}^{(0)}$, and $\Phi_m$. Set $t$ = 0.

\While{not converged}{
    Initialize $\lambda^{(0)}$, $\rho^{(0)}$, and set 
    $s = 0$, $d = 0$, $e = 0$, 
    $P_{m,n,\mathrm{inner}}^{(0)} = P_{m,n}^{(t)}$.

    \While{$\|\mathcal{N}_{r,m}^{(s+1)} - \mathcal{N}_{r,m}^{(s-1)}\|_0 > \epsilon_2$}{
        Initialize $\tilde{\mathbb{I}}_{\mathcal{N}_{r,m}}^{(0)}$.\\
        Update $\gamma_{m,n}^{(j)}$ and $\kappa_{m,n}^{(j)}$ using (\ref{eq:gamma_def}) and (\ref{eq:kappa_def}).\\
        Solve the convex subproblem (P5) via CVX.\\
        Apply (\ref{eq:thresholding}) to recover discrete allocation $\mathcal{N}_{r,m}^{(j)}$.\\
        $s \leftarrow s + 1$.
    }

    \While{$\|P_{m,n}^{(d)} - P_{m,n}^{(d-1)}\|_F^2 > \epsilon_1$}{
        Update $P_{m,n,\mathrm{inner}}^{(d+1)}$ using the water‐filling solution in (\ref{eq:optimal_power_allocation}).\\
        Update $\lambda$ via (\ref{lambda}), $\rho$ via (\ref{o}).\\
        $d \leftarrow d + 1$.
    }

    \While{$\|\phi_m^{(e1)} - \phi_m^{(e-1)}\|_F^2 > \epsilon_3$}{
        Transform $b_{k,m}$ to $x_{k,m}$ using \eqref{xtransform}.\\
        Relax $x$ via \eqref{xRelax} to form nonconvex problem (P7).\\
        Solve $P7$ in CVX.\\
        Map $x$ back to $\{0,\pi\}$ via the smooth threshold in \eqref{eq:phase_map}.\\
        $e \leftarrow e + 1$.
    }

    $t \leftarrow t + 1$.
}

\Return{$P_{m,n}^* = P_{m,n}^{(t)}$,~~$\mathcal{N}_{r,m}^* = \mathcal{N}_{r,m}^{(t)}$,$\Phi_m^*$ = $\Phi_m^{(e)}$.}
\end{algorithm}

Directly maximizing $I_r(\mathbf{x})$ under $C(\mathbf{x})\ge \Gamma_c$ and the discrete constraint $\mathbf{x}\in\{\pm1\}^K$ leads to a nonconvex integer program.  One approach is to first relax the $\{\pm1\}$ constraint into a convex semidefinite set. Specifically, we lift $\mathbf{x}$ into the matrix $\mathbf{X}=\mathbf{x}\,\mathbf{x}^\mathrm{H}\succeq0$ with diagonal entries $\mathbf{X}_{ii}=1$. For a linear or quadratic form in $\mathbf{x}$, we have
\begin{equation}
\mathbf{x}^\mathrm{H}\,\mathbf{Q}\,\mathbf{x}
\;=\;
\mathrm{Tr}(\mathbf{Q}\,\mathbf{x}\,\mathbf{x}^\mathrm{H})
\;=\;
\mathrm{Tr}(\mathbf{Q}\,\mathbf{X}). \label{xRelax}
\end{equation}

Hence, a simplified version of the relaxed problem (P5) can be reformulated as:

\vspace{-0.3cm}
\begin{subequations} \label{SDRprob}
\begin{align}
\text{(P7)} \quad \max_{\mathbf{X}\succeq0}\quad &
\sum_{m}\mathrm{Tr}\bigl(\mathbf{Q}_r^{(m)}\,\mathbf{X}\bigr)
\\[-3pt]
\text{s.t.}\quad &
\sum_{m}\mathrm{Tr} \bigl(\mathbf{Q}_c^{(m)}\,\mathbf{X}\bigr)
\;\ge\;
\Gamma_c,
\\[-3pt]
&\mathbf{X}_{ii} \;=\; 1,\quad \forall i.
\end{align}
\end{subequations}
This is a standard semidefinite program (SDP) that can be solved via convex optimization packages (e.g.\ CVX). After solving, each real-valued solution $x_k$ is mapped back onto $\{0,\pi\}$ via a smooth threshold function, e.g.

\begin{equation}
x_k
~=~
\operatorname{sign}\bigl([\mathbf v]_k\bigr),\qquad
\phi_k
~=~
\frac{\pi}{2}\Bigl(1-x_k\Bigr),
\label{eq:phase_map}
\end{equation}
where $\mathbf v$ is the unit-norm dominant eigenvector of the SDP
solution $\mathbf X^\star$.%
\footnote{If $\operatorname{rank}(\mathbf X^\star)>1$, we apply the
Gaussian randomisation technique: generate
$N_{\text{rand}}$ samples
$\tilde{\mathbf x}^{(i)}\!\sim\!\mathcal N(\mathbf 0,\mathbf X^\star)$,
quantise each sample via~\eqref{eq:phase_map}, and keep the one that
maximises the objective in~(P5).}
To prevent large oscillations at the early iterations, we first use a smooth surrogate $\widetilde{x}_k^{(t)}
~=~
\tanh \bigl(\beta_t\,[\mathbf v]_k\bigr)$, which realises deterministic annealing by gradually increasing the
slope $\beta_t$.  As $\beta_t\!\to\!\infty$, the smooth mapping
converges to the hard threshold in~\eqref{eq:phase_map}, thus providing both numerical stability and a valid binary phase solution.

\vspace{-0.3cm}
\subsection{Algorithm Complexity Analysis}
In our work, we employ a BCD algorithm (Algorithm~\ref{alg:ISAC_Alternating_Optimization_Revised}) to solve the joint optimization problem (P1) by alternately updating the RE allocation, power allocation, and BD modulation decisions. The overall complexity primarily arises from these three inner steps and the outer iteration loop.
For RE allocation, solving the convex relaxation via CVX incurs $\mathcal{O}\left((NM)^{3.5}\right)$ per iteration due to $NM$ variables and constraints. The $\ell_0$-norm approximation and parameter updates add $\mathcal{O}(NM)$ operations. With $T_1$ iterations, the RE allocation complexity is $\mathcal{O}\left(T_1 (NM)^{3.5}\right)$.
For power allocation, the closed-form water-filling solution and gradient updates contribute $\mathcal{O}(NM)$ per iteration. Given $T_2$ iterations, the complexity is $\mathcal{O}\left(T_2 NM\right)$.
For BD modulation, complexity mainly arises from solving the semidefinite relaxation problem (P7) with a $K\times M$ lifted matrix, leading to $\mathcal{O}\left((KM)^{3.5}\right)$ per iteration. The thresholding step has negligible complexity $\mathcal{O}(KM)$. With $T_3$ iterations, the BD modulation complexity totals $\mathcal{O}\left(T_3 (KM)^{3.5}\right)$.
Therefore, the total complexity across $T_4$ outer iterations is dominated by these three subproblems, 

\begin{equation}\label{eq:overall_complexity}
  \mathcal{C}_{\text{totol}}
  =\mathcal{O}\!\Bigl(
      T_4\bigl(
        T_1 (NM)^{3.5}
        + T_2\, NM
        + T_3 (KM)^{3.5}
      \bigr)
    \Bigr).
\end{equation}

The proposed algorithm thus exhibits polynomial-time complexity, which is suitable for practical ISAC systems.

\subsection{Discussion} \label{discussion}
In this section, we discuss how the proposed Pareto optimization algorithm can be adapted to more general ISAC scenarios. We consider (i) a bistatic ISAC system, where the transmitter and radar receiver are deployed at different locations, and (ii) a MIMO setting, where the BS and the BD employ multiple antennas.

\subsubsection{Adaptability for Bistatic ISAC Systems}

In the bistatic setting, the user side remains the same, while the radar transmitter and receiver are spatially separated. Rather than having a monostatic echo channel, the radar channel in \eqref{radar receiver} now becomes:

\vspace{-0.7cm}
\begin{align}
 Y_{r}(m,n) = \alpha_t (&H_{s,t}H_{t,r}   + 
 \sum_{k=1}^K b_{k,m}H_{s,t}H_{t,k}H_{k,r}) \\  \nonumber
 &X(m,n)   + V_{r}(m,n),
\end{align}
where $H_{s,t}$, $H_{t,r}$, $H_{k,r}$ represent the frequency channel response of the BS to the target, the target to the sensing receiver, and the k-th BD to the sensing receiver, respectively.

Compared with the SMI equation of the proposed monostatic ISAC system defined in \eqref{radar receiver}, the only difference is the change of the cascade channel response. While this changes how the SMI is mathematically expressed, the structure of the RE allocation, power allocation, and BD phase modulation optimization subproblems remains essentially the same, while constraints such as total power, subcarrier allocation, and BD phase shifts remain unchanged but use bistatic channel parameters. Thus, the overall structure of Algorithm \ref{alg:ISAC_Alternating_Optimization_Revised} is preserved, directly extends to bistatic scenarios with suitable adjustments to channel modeling, though numerical outcomes may vary due to geometric differences.  

However, a bistatic ISAC system may require more careful synchronization between the BS and the sensing receiver. If they are not co-located, one must ensure that OFDM or other waveforms remain time/frequency synchronized at both ends for coherent processing \cite{10649809}. In addition, cross-interference or pilot contamination can arise if the BS and the radar receiver share spectrum but do not share a common clock. Appropriate pilot design and interference coordination strategies can address this \cite{yang2017modulation,
xie2024integration}.

\subsubsection{Adaptability to MIMO Settings}

A remarkable advantage of MIMO systems lies in utilizing multiple antennas at the transmitter and/or receiver to enhance both the communication and radar functionalities. For instance, consider a MIMO scenario with a BS equipped with $M_a$ antennas and a number of BDs, each equipped with $M_b$ antennas \cite{9778563}. Suppose there are $N$ subcarriers indexed by $n\in\{1,\ldots,N\}$. Then, each channel response in \eqref{channel} becomes a three-dimensional channel matrix
\begin{equation}
  \mathcal{H} 
  \;=\;
  \bigl\{ 
    h_{n,m_a,m_b}\;\big|\; 
    n\!\in\![1,N],\, 
    m_{a}\!\in\![1,M_{a}],\,
    m_{b}\!\in\![1,M_{b}]
  \bigr\},
\end{equation}
where $h_{n,m_a,m_b}$ represents the fading coefficient (e.g., channel impulse response) on subcarrier $f_n$ between the $m_a$-th antenna at the BS and the $m_b$-th antenna at the $k$-th BD.  
This will bring benefits to both communication and radar sensing performance. Specifically, multi-antenna arrays can enhance target detection and parameter estimation (e.g., angle-of-arrival \cite{yu2022location}), thus increasing $I_r$ for the same power budget. Moreover, MIMO spatial multiplexing boosts the data rate $C_d$, enabling better performance under the same total power. With additional spatial degrees of freedom, the achievable Pareto boundaries tend to shift outward, offering more favorable simultaneous sensing and communication performance. In addition, beamforming strategies could be exploited to boost the performance further. Nevertheless, this paper mainly focuses on the time-frequency optimization of the multiple carrier system; the beamforming strategy can refer to some existing schemes in references \cite{luo2024isac,he2023full}. 

In addition, the algorithmic nature of optimization has not fundamentally changed. Whether it is the resource allocation, the design of power, or the binary optimisation of BD reflection phases, they can still be solved iteratively by the same BCD framework. Their subproblems can usually still be solved with the help of convex approximations (e.g. SDR, SCA, etc.) or analytical updates (e.g. gradient-based water-filling deformation). In other words, the overall idea of the algorithm remains the same as in the single-antenna case, with the main difference being that the system becomes a high-dimensional matrix variable and the corresponding constraints are more complex, but it does not change the structure of the optimization problem and the BCD iteration mechanism.




\begin{table}[t]
\centering
\caption{Default System Parameters}
\label{tab:default_system_params}
\resizebox{\linewidth}{!}{
\begin{tabular}{l l}
\hline
\textbf{System parameters} & \textbf{Value} \\
\hline
Transmit Power ($P_t$)                         & 0 dBm \\
Number of BDs ($K$)                      &50              \\
Carrier frequency  ($f_c$)                      & 28\,GHz        \\
 Number of subcarriers (\(N\))  & 128 \\  
The distance between BS and Target (\(D_{S,T}\))   & \(8\) $\mathrm{m}$ \\
   The distance between BS and User   (\(d_{B,U}\))  & \(10\) $\mathrm{m}$ \\
        The activity range of $K$ BDs around other devices (\(D_{k}\))& 0.1  -  0.5 $\mathrm{m}$ \\

      OFDM symbol duration   (\(T\))  & \(4.1470 \, \mathrm{us}\) \\  
      Backscatter coefficient   ($\alpha_k$)& 0.5\\ 
    Channel fading    ($H_{\alpha, \beta}(n,m)$)& Rayleigh fading\\ 
      
\hline
\end{tabular}
}
\end{table}

\section{Performance Evaluation} \label{performance}

In this section, we describe the experimental setup and evaluate the proposed ISAC system regarding the analyzed metrics, i.e., SMI and data rate. Specifically, we present the Pareto boundary under various parameters, analyze sensing and communication performances individually, and compare our proposed system against state-of-the-art benchmarks regarding the ISAC performance and cost.

\vspace{-0.3cm}
\subsection{Simulation Setup}

We consider a BD-assisted ISAC system comprising a BS, multiple BDs, a user, and a sensing target. In line with the system model described in Section \ref{System model}, the BS simultaneously serves as a transmitter and sensing receiver, transmitting OFDM samples to communicate with the user and to sense the target. The BDs, naturally distributed around typical ISAC entities such as the BS, user, or target, reflect the incoming OFDM signals within a certain activity range to these entities.
We adopt practical system parameters from existing studies \cite{10221890,10566596,bicua2019multicarrier} in the related scenario, which are summarized in Table \ref{tab:default_system_params}. In addition, to demonstrate the sensing and communication performance individually, we use the proposed algorithm to solve both the sensing-centric optimization problem (P1) defined in \eqref{P1} and use a similar algorithm where the difference is that it regards the SMI as the constraints to solve the communication-centric optimization problem (P2) defined in \eqref{P2}.

We simulate the proposed BD-assisted system with the optimization for RE (i.e., subcarrier), power, and phase of BDs, denoted as ``\textbf{Proposed, SPP}". This corresponds to the Scenario 1 defined in Section \ref{Backscatter Model}.
In addition, we also simulate a variant of our proposed system optimized solely for subcarrier and power, excluding the optimization of BDs' phase, labeled as ``\textbf{Proposed, SP}". This scheme works effectively in fast fading scenarios where precise BD operation control is unavailable, which corresponds to the Scenario 2 defined in Section \ref{Backscatter Model}. To demonstrate the superiority and effectiveness of the proposed system, we compare it with three benchmark schemes.
First, we include a state-of-the-art benchmark from reference \cite{bicua2019multicarrier}, labeled as ``\textbf{Ref.} \cite{bicua2019multicarrier}." Similar to our approach, this scheme analyzes the trade-off between sensing and data rates in an OFDM-based ISAC system. However, it differs by not considering BDs.
Second, we simulate a benchmark that uses RIS as the assisted device in the ISAC system, labeled as ``\textbf{RIS-based}. To ensure fairness, the RIS is equipped with $K$ reflection elements, which is the same to the total BD number. In addition, it still use the proposed algorithm that optimize the RE and power, while the phase for the RIS is optimized using the interior-point method, which is a commercial Matlab algorithm frequently used for RIS optimization \cite{mathworks_ris_cdl}. 
Third, a scheme that simply employs time division multiple access (TDMA) and frequency division multiple access (FDMA) for sensing and communication resource allocation without optimization, labeled as ``\textbf{TDMA/FDMA}," is established to specifically evaluate the effectiveness of our proposed optimization algorithm.
By systematically comparing our proposed system and algorithm against these three distinct benchmarks, we comprehensively demonstrate the superior performance and practical advantages of the proposed BD-assisted ISAC system.

\begin{figure}
\centering
\subfigure[The convergence of SMI.]{
    \includegraphics[width=0.8\linewidth]{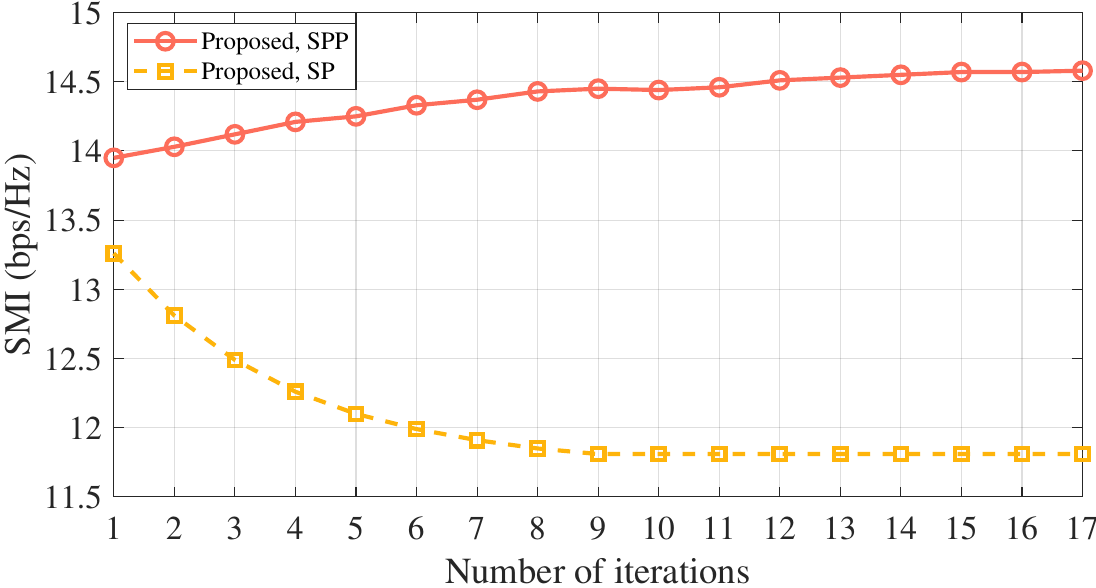} 
    \label{con1} 
    }  
\subfigure[The convergence of data rate.]{
    \includegraphics[width=0.8\linewidth]{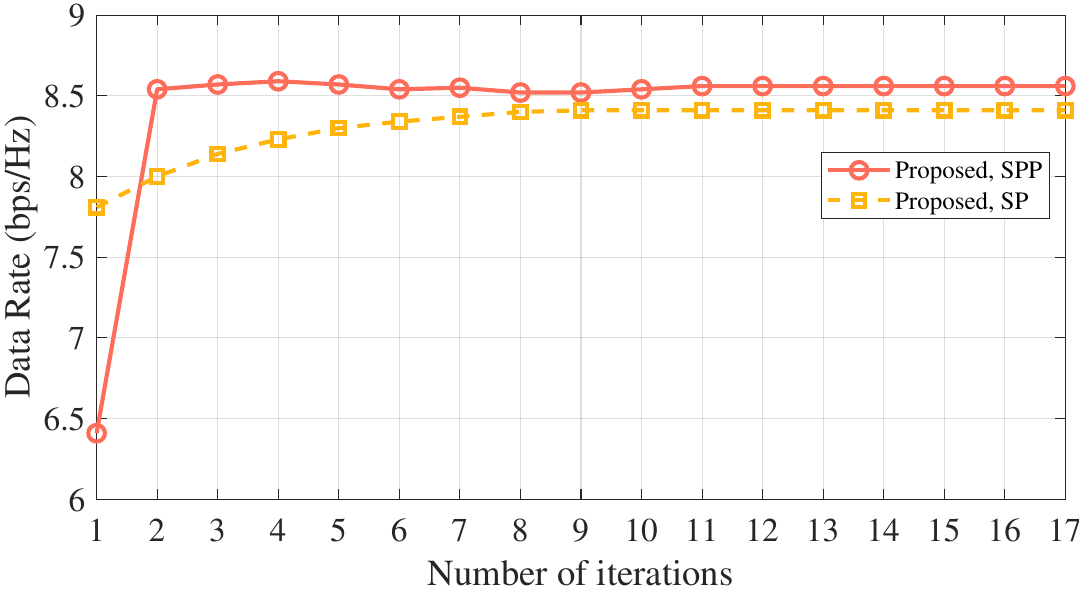}
    \label{con2}
    }
\caption{Convergence behaviors of the proposed algorithm.}
\label{con}
\vspace{-0.3cm}
\end{figure}

\vspace{-0.4cm}
\subsection{Simulation Results}

Fig. \ref{con} illustrates the convergence behavior of SMI and data rate for the proposed system, where the communication constraint $\Gamma_c$ and the sensing constraint $\Gamma_s$ are set to 14.5 bps/Hz and 8.5 bps/Hz, respectively. As observed from Fig. \ref{con1}, the proposed system's SMI quickly and monotonically converges to a stable value after approximately 6 iterations, demonstrating the algorithm's efficiency and its convergence. In contrast, the baseline without BDs phase optimization exhibits a lower sensing performance, demonstrating the effectiveness of the proposed system in effectively optimizing BDs modulation decisions.
Similarly, Fig. \ref{con2} shows the convergence of data rate that rapidly converges within a few iterations, indicating the convergence of the proposed algorithm in addressing the communication-centric problem. Both schemes, after an initial rapid increase, gradually converge to a stable value. These results confirm the rapid convergence of the proposed algorithm and the superior performance of the proposed ``\textbf{Proposed, SPP}" algorithm than the proposed ``\textbf{Proposed, SP}" algorithm.

\begin{figure*}[!tp]
\centering
\begin{minipage}[b]{0.32\linewidth}
  \centering
  \includegraphics[width=\linewidth]{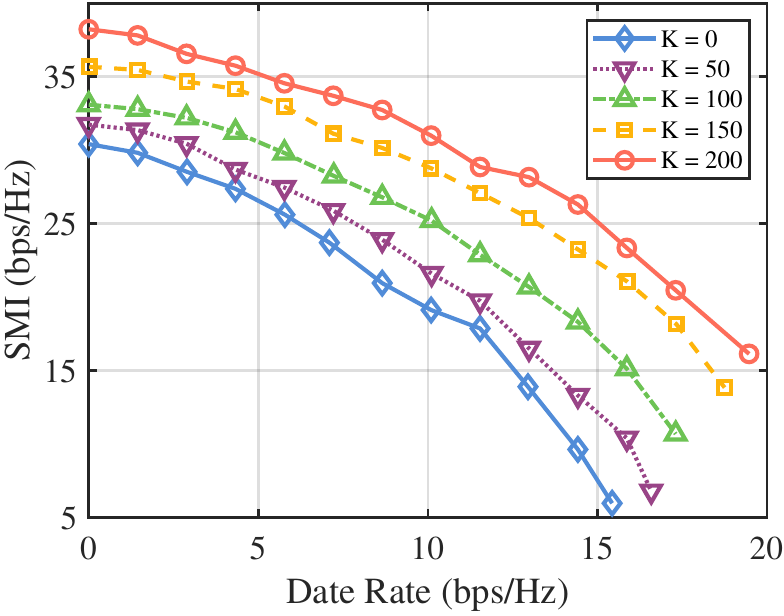}
  \vspace{-2em}
  \caption{Pareto boundaries for different number of BDs.}
  \label{1}
\end{minipage}
\hfill
\hspace{-3mm}
\begin{minipage}[b]{0.32\linewidth}
  \centering
  \includegraphics[width=\linewidth]{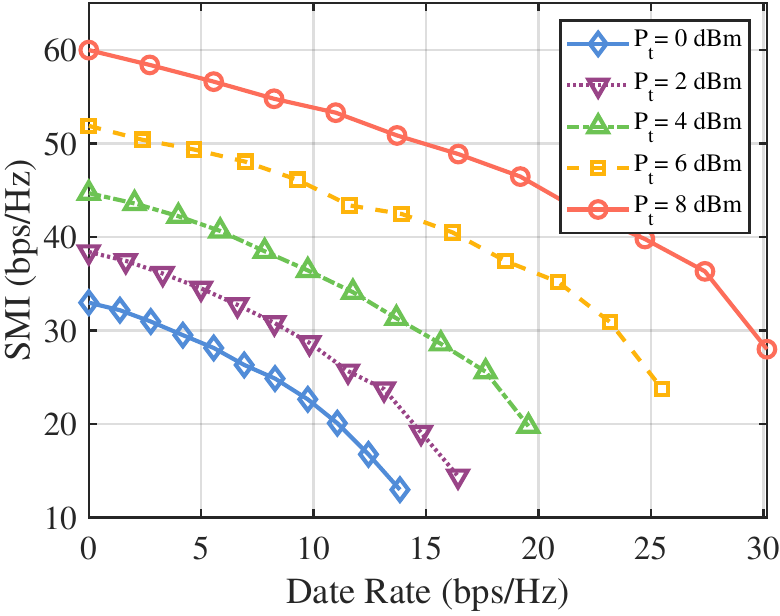}
    \vspace{-2em}
  \caption{Pareto boundaries for different power budgets.}
  \label{2}
\end{minipage}
\hfill
\hspace{-3mm}
\begin{minipage}[b]{0.32\linewidth}
  \centering
  \includegraphics[width=\linewidth]{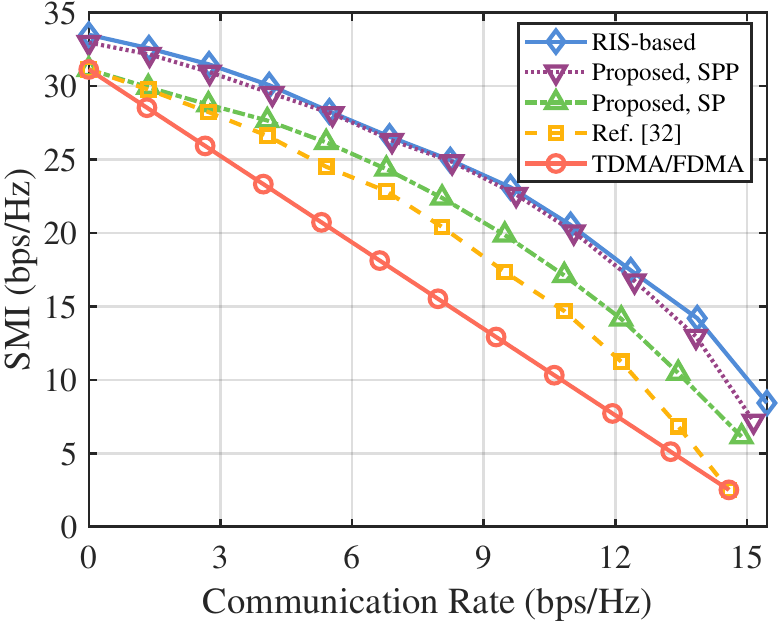}
    \vspace{-2em}
  \caption{Pareto boundaries for different schemes.}
  \label{3}
\end{minipage}
\vspace{-0.5cm}
\end{figure*}

\begin{figure}
\centering
\subfigure[SMI versus different transmit power.]{
    \includegraphics[width=0.8\linewidth]{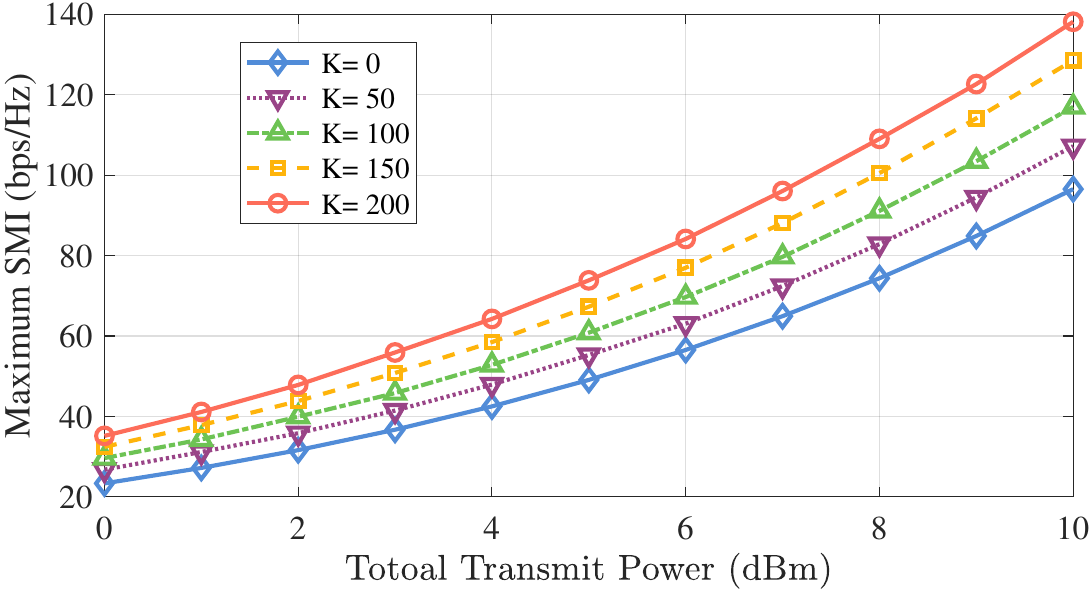} 
    \label{4.1} 
    }  
\subfigure[Data rate versus different transmit power.]{
    \includegraphics[width=0.8\linewidth]{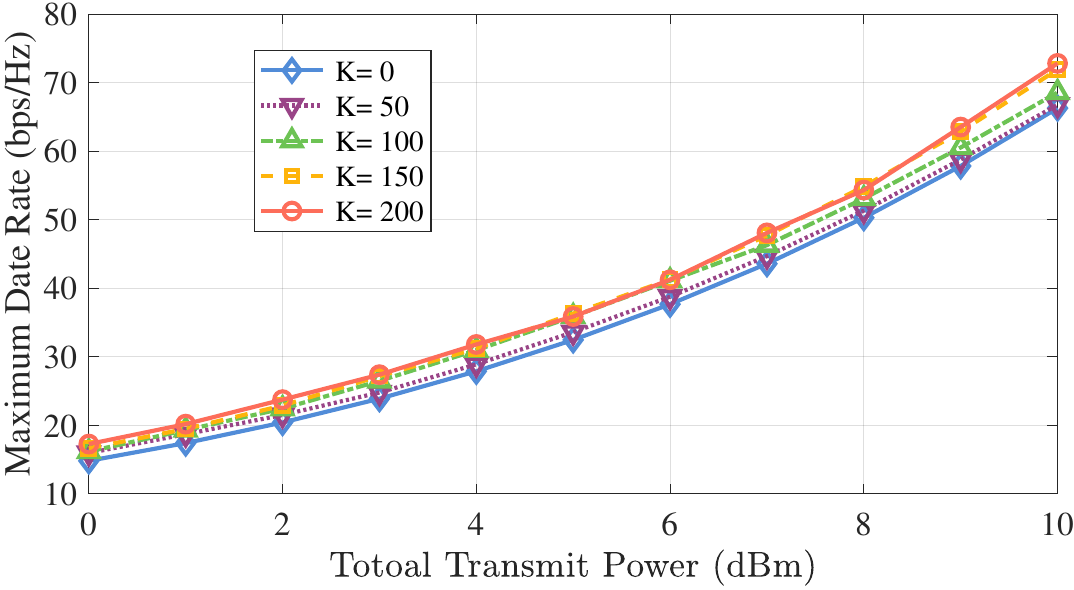}
    \label{4.2}
    }
\caption{Sensing and communication performance under various numbers of BDs.}
\label{4}
\vspace{-0.5cm}
\end{figure}

\begin{figure}
\centering
\subfigure[SMI vs. transmit power under various BD activity range]{
    \includegraphics[width=0.8\linewidth]{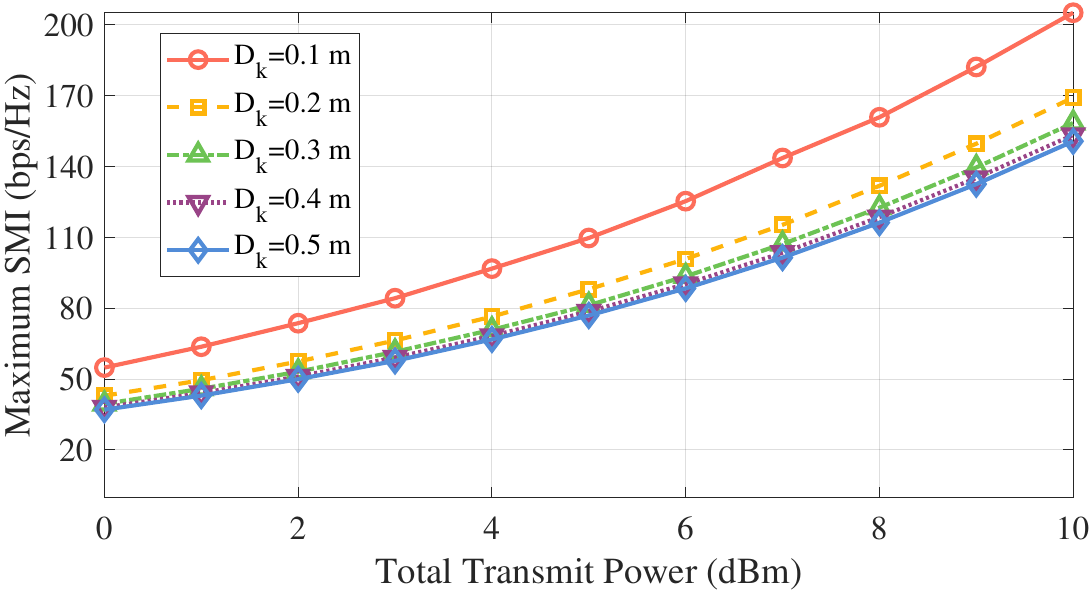} 
    \label{5.1} 
    }  
\subfigure[Data rate vs. transmit power under various BD activity range]{
    \includegraphics[width=0.8\linewidth]{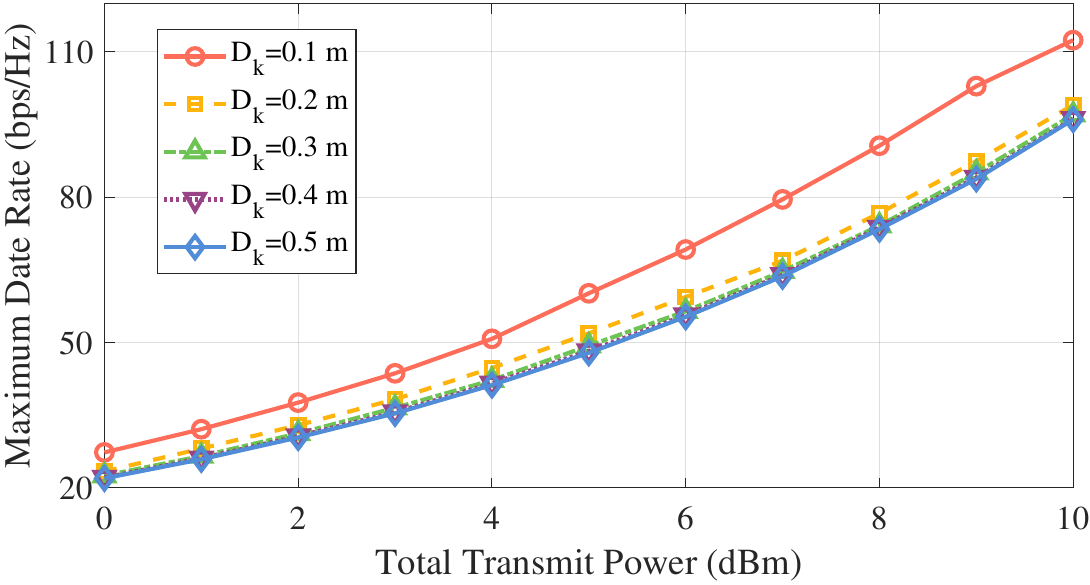}
    \label{5.2}
    }
\caption{Sensing and communication performance under various BD activity ranges.}
\label{PLA_BC}
\vspace{-0.5cm}
\end{figure}
Fig. \ref{1} illustrates the Pareto boundary showing the trade-off between the SMI and the data rate for different numbers of BDs. As observed, the proposed BD-assisted ISAC system consistently achieves improved performance with the growth of the number of
BDs. Specifically, increasing the number of BDs from 0 to 200 substantially enhances the achievable SMI for any given data rate, demonstrating the effectiveness of utilizing backscattered signals from additional BDs. Furthermore, the SMI gradually decreases as the required data rate increases, which clearly indicates the inherent trade-off in the joint allocation of limited resources between sensing and communication functions.



Fig. \ref{2} illustrates the impact of varying transmit power budgets on the Pareto boundary between SMI and data rate. As clearly shown, increasing the transmit power from 0 dBm to 8 dBm significantly enhances the achievable SMI for a given communication requirement. This result indicates that higher transmit power improves the reflected signal strength, thus directly augmenting sensing capabilities. Additionally, it is evident that as the data rate constraint increases, the achievable SMI gradually declines, demonstrating the inherent resource allocation trade-off between sensing and communication functionalities in the proposed BD-assisted ISAC system.


Fig.  \ref{3} compares the Pareto boundary of the proposed BD-assisted ISAC system against various benchmark schemes. It is clearly shown that the proposed system significantly outperforms all benchmarks. Specifically, it is approximately 15\% improvements in both the SMI and the data rate between ``\textbf{Proposed, SPP}" and ``\textbf{Ref. [32]}".
This performance advantage arises from the effective utilization of communication subcarriers for significant gains and the reflected signals from backscatter devices, thereby enhancing both sensing and communication capabilities simultaneously. The results further highlight that introducing BDs and jointly optimizing resource allocation considerably improve system performance compared to the benchmarks. In addition, although the "\textbf{RIS-based}" benchmark achieves a marginal performance improvement compared to the "\textbf{Proposed, SPP}" method, however, this gain is limited (less than 5
\%) and necessitates a more delicate RIS deployment, introducing additional costs in hardware and labor.

Fig. \ref{4} illustrates the impact of transmit power on the sensing and communication performance for varying numbers of BDs. In Fig. \ref{4.1}, it is clearly demonstrated that increasing transmit power consistently enhances the achievable SMI, particularly when more BDs are deployed. For instance, increasing the number of BDs from 0 to 200 results in a substantial improvement in SMI at all power levels, emphasizing the effectiveness of BDs in augmenting sensing performance via additional reflected signal paths.
Fig. \ref{4.2} shows the corresponding communication performance. Similarly, the data rate increases as the transmit power grows. 
These results collectively highlight the advantage of utilizing BDs in enhancing sensing capability significantly while maintaining stable data rate growth with increased transmit power.


Fig. \ref{5.1} illustrates how the SMI varies with transmit power budgets for different BD activity ranges around other devices, i.e., the BS, the target, or the user. Clearly, placing BDs closer to the BS substantially improves the SMI at identical power levels, since shorter distances between BDs and BS reduce signal reflection loss and directly enhance sensing performance.
Fig. \ref{5.2} illustrates the corresponding data rate under similar conditions. Similar to the sensing results, positioning BDs closer to devices further enhances communication performance. This is because the shorter signal propagation distance reduces the path loss. These observations emphasize the significance of strategically placing BDs closer to basis ISAC entities to maximize overall ISAC system performance.

\begin{figure*}[!tp]
\centering
\begin{minipage}[b]{0.33\linewidth}
  \centering
  \includegraphics[width=\linewidth]{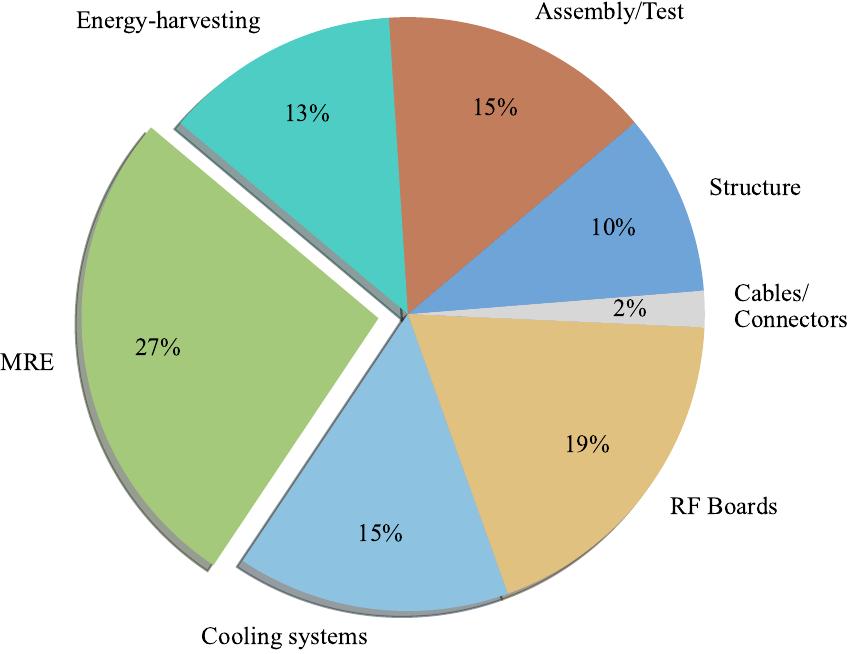}
  \vspace{-2em}
  \caption{Cost distribution of a RIS.}
  \label{Cost1}
\end{minipage}
\hfill
\hspace{-9mm}
\begin{minipage}[b]{0.33\linewidth}
  \centering
  \includegraphics[width=\linewidth]{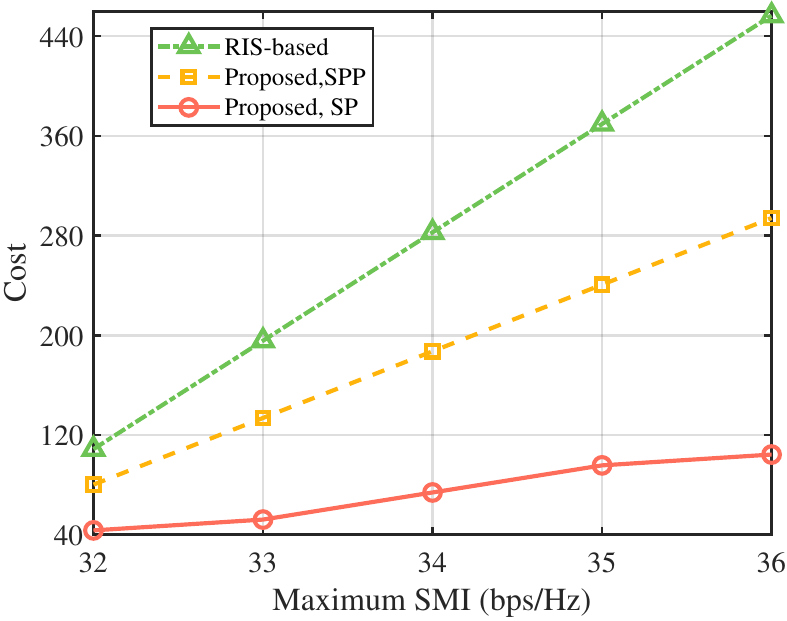}
    \vspace{-2em}
  \caption{Hardware const vs. maximum SMI}
  \label{Cost2}
\end{minipage}
\hfill
\hspace{-9mm}
\begin{minipage}[b]{0.33\linewidth}
  \centering
  \includegraphics[width=\linewidth]{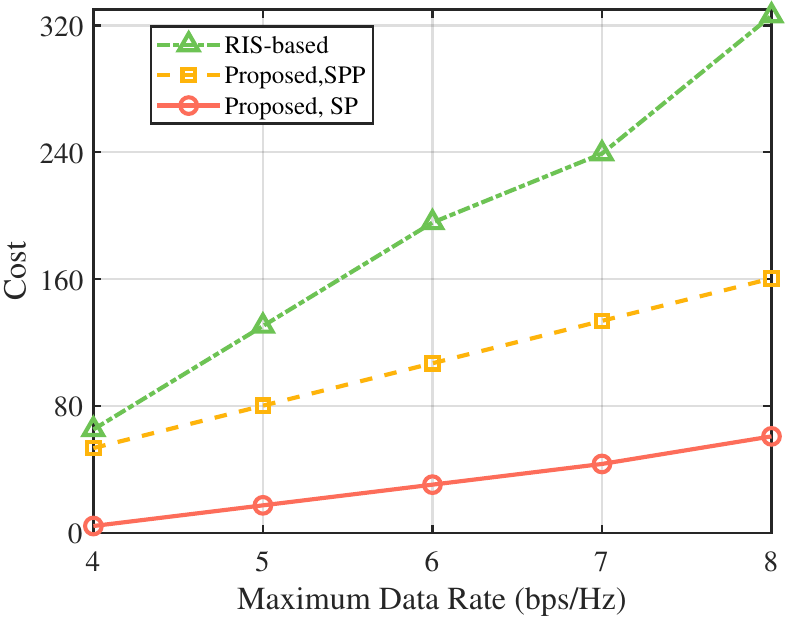}
    \vspace{-2em}
  \caption{Hardware const vs. maximum data rate}
  \label{Cost3}
\end{minipage}
\vspace{-0.3cm}
\end{figure*}

At least, we approximately evaluate the cost performance between the RIS-assisted ISAC system and our proposed system.
RIS consists of dense arrays of passive metasurface unit cells, which interface with a central controller and RF interconnects. As shown in Fig. \ref{Cost1}, the single-unit bill of materials can be normalised as metamaterial radiation element (MRE), cooling system, RF boards, cables/connectors, mechanical structure, assembly \& test, and energy-harvesting module \cite{9696209,
rossanese2023data}.  
When the same aperture is implemented as an BD, the metamaterial element is replaced by a simple antenna, the entire cooling block is removed, and the assembly \& test share is reduced to $2\,\%$.  All other blocks (RF boards, cables, structure, energy harvesting) are retained, so the remaining share sums to $50\,\%$. Let $C_0$ be the hardware cost of one metasurface element and let
$c\,(2\le c\le 10)$ be the cost ratio between a metasurface element and a simple antenna.  Then $Cost_{\mathrm{RIS}}=\frac{100}{23}\,N\,C_{0}$,  $Cost_{\mathrm{BD,SPP}}=\Bigl(\frac{1}{c}+\frac{50}{23}\Bigr)\,N\,C_{0}$, where $N$ denotes the number of elements (or BDs) under consideration. Since BDs are typically semi-passive under the SPP scheme, we set $c = 2$. For convenience, we set $C_0=1$ in numerical examples. In addition, for the SP scheme, since it completely leverages BDs already distributed in the environment, its hardware cost can be eliminated. We only consider its small assembly/test cost for approximately  $7\%$, i.e., $Cost_{\mathrm{BD,SP}}=\Bigl(\frac{7}{23}\Bigr)\,N\,C_{0}$. From Fig. \ref{Cost2} and \ref{Cost3}, it can be seen that the cost of our proposed system is significantly lower than the RIS-based scheme in the same maximum SMI and maximum data rate level, approximately $50\%$ saving in the ``\textbf{Proposed, SPP}" scheme and $80\%$ saving in the ``\textbf{Proposed, SP}" scheme. This indicates that our proposed system provides a powerful solution to reduce the hardware cost while guaranteeing a high sum-rate in practice.

\vspace{-0.1cm}
\section{Conclusion}\label{Conclusion}

This paper proposed a novel BD-assisted ISAC system. To illustrate the inherent trade-offs between sensing and communication, a Pareto optimization framework was established, formulating dual optimization problems to maximize either SMI or communication rates under mutual constraints, considering multiple resource optimization on time-frequency RE allocation, power allocation, and BD modulation decision.
To effectively address the non-convex optimization challenges, a BCD algorithm was developed, where the SCA method is proposed to solve the RE allocation problem, the augmented Lagrangian method combined with the water-filling method is proposed to solve the power allocation problem, and the SDR method is proposed to solve the BD modulation decision problem. Furthermore, we discussed the adaptability of our proposed system in bistatic ISAC and MIMO systems. Extensive simulation results demonstrated the performance of the proposed system across diverse system conditions and confirmed its notable performance improvements over traditional ISAC schemes and other benchmark methods. These results highlight the potential of leveraging BDs in practical ISAC systems for performance enhancement.

\appendices

\section{Proof of Compactness and Normality of the Achievable Performance Region} \label{appendice1}

In this appendix, we aim to prove that the achievable performance region $\mathcal{R}$ of the proposed system, defined in \eqref{region}, is both compact and normal

\vspace{-0.3cm}
\subsection{Compactness} 

A set in $\mathbb{R}_+^2$ is compact if it is closed and bounded \cite{chen2021joint}. We show each property in turn.

Since the allocated power for sensing and communication is constrained by a total transmit power budget $P_t$, all feasible power allocations satisfy:
\begin{align}
0 \leq P_c(m,n) + P_s(m,n) \leq P_t, \quad \forall m, n.
\end{align}
And we further have 
\begin{align}
0 \leq P_c(m,n) \leq P_t, \quad 0 \leq P_s(m,n) \leq P_t, \quad \forall m, n.
\end{align}

In addition, for subcarrier allocation, we also have:
\begin{align}
\mathcal{N}_{r,m} \cup \mathcal{N}_{c,m} = \mathcal{N}, \;\mathcal{N}_{r,m}\cap \mathcal{N}_{c,m}=\emptyset, \quad \forall m, n.
\end{align}
Therefore, all feasible points $(I_s,C_d)\in \mathcal{R}$ arise from a finite power budget $P_t$ and a finite set of subcarriers $\mathcal{N}$ with maximum per-resource-element power constraint $P_{\text{max}}$. Hence, there exists an intrinsic upper bound on the maximum power usage. Since $I_s$ is monotonically non-decreasing with transmit power devoted to radar subcarriers and $C_d$ is monotonically non-decreasing with transmit power devoted to communication subcarriers, both $I_s$ and $C_d$ are upper-bounded by the maximum feasible power allocation (i.e., $P_t$ or total subcarrier usage). Thus, there exist finite constants $I_s^{\max}$ and $C_d^{\max}$ such that 
\begin{equation}
0 \;\le\; I_s \;\le\; I_s^{\max}, 
\quad\text{and}\quad 
0 \;\le\; C_d \;\le\; C_d^{\max},
\end{equation}
for all $(I_s,C_d)\in \mathcal{R}$. Consequently, $\mathcal{R}$ is bounded in $\mathbb{R}_+^2$.

Furthermore, to show that $\mathcal{R}$ is closed, consider any convergent sequence $\bigl\{(I_r^{(i)},\,C_d^{(i)})\bigr\}_{i=1}^\infty \subseteq \mathcal{R}$.
Each pair $(I_r^{(i)},C_d^{(i)})$ is generated by some feasible finite-dimensional set of variables $\bigl\{ P_{n,m}^{(i)},\,\phi_{k,m}^{(i)},\,\mathcal{N}_{r,m}^{(i)} \bigr\}$. Since (i)~$P_{n,m}^{(i)} \in [0, P_{\max}]$ and $\sum_{n,m}P_{n,m}^{(i)}\leq P_t$, (ii)~$\phi_{k,m}^{(i)}\in \{0,\pi\}$, (iii)~$\mathcal{N}_{r,m}^{(i)}$ and $\mathcal{N}_{c,m}^{(i)}$ each come from a finite set of possibilities, the space of such resource assignments is compact being the finite Cartesian product of intervals or discrete sets. Consequently, there exists a convergent subsequence whose resource allocations converge to a limit optimal vector $\bigl\{P_{n,m}^*, \phi_{k,m}^*, \mathcal{N}_{r,m}^*\bigr\}$, which still must satisfy the same constraints. By the continuity of $I_r$ and $C_d$ with respect to the channels, power allocation, and phases, it follows that
\begin{equation}
\lim_{i\to \infty} (I_r^{(i)}, C_d^{(i)}) 
\;=\; (I_r^*, C_d^*),
\end{equation}
and $(I_r^*, C_d^*)$ continues to satisfy all constraints, thus $(I_r^*,C_d^*)\in \mathcal{R}$. Hence, every limit point of $\mathcal{R}$ remains in $\mathcal{R}$, proving $\mathcal{R}$ is closed.

\subsection{Normality}

The set $\mathcal{R}\subseteq\mathbb{R}^2_+$ is normal if, for any point $(I_r, C_d) \in \mathcal{R}$ that for $(I_r',C_d')$ satisfying $0 \le I_r' \le I_r$ and $0 \le C_d' \le C_d$, it must hold that $(I_r', C_d') \in \mathcal{R}$.

To prove that $\mathcal{R}$ is normal, suppose we have a point $(I_s,\, C_d)\in \mathcal{R}$. By definition, certain power allocations, subcarrier allocations, and BD modulation decisions exist, such that the resulting SMI is $ I_r$ and the data rate is $C_d$. Now consider any pair $(I_s', C_d')$ with
\begin{equation}
0 \;\le\; I_s' \;\le\; I_s, 
\quad\text{and}\quad 
0 \;\le\; C_d' \;\le\; C_d.
\end{equation}
We can always obtain $(I_s',C_d')$ by, for instance, reducing the power on specific subcarriers or subcarrier sets to degrade the SMI or the communication rate intentionally. For example, if we only partially use or randomly modulate some subcarriers (thereby lowering the effective energy for radar sensing), $I_s$ can be reduced to $I_s'$. Similarly, by disabling or reducing power on the communication subcarriers, $C_d$ can be reduced to $C_d'$.
Such partial or “throttled” usage still satisfies all feasibility constraints (e.g., total power remains $\le P_t$, subcarrier assignments remain valid, BD phase shifts remain feasible). Therefore, $(I_s',C_d')$ remains in the feasible region, implying $\mathcal{R}$ is downward closed in the positive quadrant. Hence, $\mathcal{R}$ is a normal set.

Combining the above, we conclude that $\mathcal{R}$ is both compact and normal.

\vspace{-0.3cm}
\section{Proof of Lemma III.2} \label{appendice2}

First, for a given $\Gamma_c$, the optimal solution to \eqref{P1} that maximizes $I_r$ in \eqref{eq:obj} subject to $C_d \ge \Gamma_c$ in \eqref{eq:comm_rate} cannot be strictly outperformed in both objectives, due to the normality and compactness of the feasible region proved in Lemma~1. Hence, it belongs to the Pareto boundary. Second, any Pareto-optimal point $(I_r^*, C_d^*)$ is retrievable by setting $\Gamma_c = C_d^*$. If there existed $(I_r', C_d')$ with $I_r' \ge I_r^*$ and $C_d' \ge C_d^*$, it would contradict the Pareto-optimality of $(I_r^*, C_d^*)$. A similar argument holds for \eqref{P2}, where a sensing constraint $I_r \ge \Gamma_s$ \eqref{14b} is imposed to maximize $C_d$ \eqref{14a}. Finally, by varying $\Gamma_c$ (or $\Gamma_s$) over the entire feasible range, one enumerates all optimal solutions, thereby obtaining the full Pareto boundary.

\section{Derivation of \eqref{gradient updates1} and \eqref{gradient updates2}} \label{appendice3}

From \eqref{eq:Ir_again}, Recall
\begin{equation}
I_r(m,n) \;=\; \Delta f \,\log_2\!\Bigl(1 \;+\; \tfrac{|\alpha_t\,G\,G^H|^2\,P(m,n)}{\sigma_r^2}\Bigr), \label{AppC-1}
\end{equation}

Since $\log_2(1 + x) \;=\;\frac{\ln(1 + x)}{\ln 2}$, \eqref{AppC-1} can be reformulated as

\begin{equation}
I_r(m,n) 
\;=\;
\Delta f \,\frac{1}{\ln 2}\,\ln\!\Bigl(1 + \tfrac{|\alpha_t\,G\,G^H|^2\,P(m,n)}{\sigma_r^2}\Bigr).\label{AppC-2}
\end{equation}

Define $x \;=\; \frac{|\alpha_t\,G\,G^H|^2\,P(m,n)}{\sigma_r^2}$. Then, we have
\begin{equation}
\frac{d}{dP(m,n)}\,\bigl[\ln(1 + x)\bigr]
\;=\;
\frac{1}{1 + x}\;\frac{d x}{dP(m,n)}. \label{AppC-3}
\end{equation}
where 
$\frac{d x}{dP(m,n)}
\;=\;
\frac{|\alpha_t\,G\,G^H|^2}{\sigma_r^2}.$

Substituting \eqref{AppC-3} into \eqref{AppC-2}, we have
\begin{align}
\frac{\partial I_r(m,n)}{\partial P(m,n)}
&=\;
\Delta f \,\frac{1}{\ln 2}
\;\frac{1}{1 + \tfrac{|\alpha_t\,G\,G^H|^2P(m,n)}{\sigma_r^2}}
\;\frac{|\alpha_t\,G\,G^H|^2}{\sigma_r^2} \nonumber
\\[4pt]
&=\;
\Delta f\,\frac{1}{\ln 2}\,
\frac{|\alpha_t\,G\,G^H|^2}{\sigma_r^2 + |\alpha_t\,G\,G^H|^2\,P(m,n)},
\end{align}
which matches the expression in \eqref{gradient updates1} of the main text. Following a similar prceudre, we also can derive

\begin{align}
\frac{\partial C_d(m,n)}{\partial P(m,n)}
&=\;
\Delta f\,\frac{1}{\ln 2}\,
\frac{|H_c|^2}{\sigma_c^2 + |H_c|^2\,P(m,n)},
\end{align}
which is precisely \eqref{gradient updates2} in the main text.

\vspace{-0.3cm}


\bibliographystyle{IEEEtran}
\bibliography{IEEEabrv,reference}

\end{document}